\documentclass{style/sig-alternate-05-2015}

\usepackage{balance}  
\toappear{}

\usepackage{graphicx}
\usepackage{balance}  
\usepackage{booktabs}
\usepackage{amsmath}
\usepackage{verbatim}
\usepackage{algorithmicx}
\usepackage[plain]{algorithm}
\usepackage{algpseudocode}
\usepackage{pgfplots}
\usepackage{pgfplotstable}
\usepackage{tikz}
\usepackage[center]{subfigure}
\usepackage{url}

\pgfplotsset{every linear axis/.append style={width=4.5cm, height=3.5cm,ylabel near ticks, xlabel near ticks}}
\usetikzlibrary{patterns,shapes,arrows,positioning,shapes.geometric}
\usepgfplotslibrary{groupplots}

\newlength{\abovecaptionskip}
\newfont{\mycrnotice}{ptmr8t at 7pt}
\newfont{\myconfname}{ptmri8t at 7pt}
\let\crnotice\mycrnotice%
\let\confname\myconfname%

\title{A Fast Randomized Algorithm for Multi-Objective Query Optimization}

\begin{document}


\newtheorem{lemma}{Lemma}
\newtheorem{theorem}{Theorem}
\newtheorem{corollary}{Corollary}
\newtheorem{example}{Example}
\newtheorem{assumption}{Assumption}
\newtheorem{definition}{Definition}



%
%
%
%
\numberofauthors{1} 

\author{
%
%
\alignauthor
Immanuel Trummer and Christoph Koch\\
			 \email{\{firstname\}.\{lastname\}@epfl.ch}\\
      \affaddr{\'Ecole Polytechnique F\'ed\'erale de Lausanne}
}

\urlstyle{same}

\maketitle

\begin{abstract}
Query plans are compared according to multiple cost metrics in multi-objective query optimization. The goal is to find the set of Pareto plans realizing optimal cost tradeoffs for a given query. So far, only algorithms with exponential complexity in the number of query tables have been proposed for multi-objective query optimization. In this work, we present the first algorithm with polynomial complexity in the query size. 

Our algorithm is randomized and iterative. It improves query plans via a multi-objective version of hill climbing that applies multiple transformations in each climbing step for maximal efficiency. Based on a locally optimal plan, we approximate the Pareto plan set within the restricted space of plans with similar join orders. We maintain a cache of Pareto-optimal plans for each potentially useful intermediate result to share partial plans that were discovered in different iterations. We show that each iteration of our algorithm performs in expected polynomial time based on an analysis of the expected path length between a random plan and local optima reached by hill climbing. We experimentally show that our algorithm can optimize queries with hundreds of tables and outperforms other randomized algorithms such as the NSGA-II genetic algorithm over a wide range of scenarios. 
\end{abstract}


\section{Introduction}

Multi-objective query optimization compares query plans according to multiple cost metrics. This is required to model scenarios such as cloud computing where users care about execution time and monetary fees for cloud resources~\cite{Trummer2014}. Another example is approximate query processing where users care about execution time and result precision~\cite{Agarwal2012}. The goal of multi-objective query optimization is then to find the set of Pareto-optimal plans, i.e.\ the query plans realizing optimal cost tradeoffs for a given query. The optimal cost tradeoffs can either be visualized to the user for a manual selection~\cite{Trummer2015a} or the best plan can be selected automatically out of that set based on a specification of user preferences (i.e., in the form of cost weights and cost bounds~\cite{Trummer2014}).

So far, exhaustive algorithms~\cite{ganguly1992query, Trummer2015} and several approximation schemes~\cite{Trummer2014, Trummer2015a} have been proposed to solve the generic multi-objective query optimization problem. The exhaustive algorithms formally guarantee to find the full Pareto frontier while the approximation schemes formally guarantee to approximate the Pareto frontier with a certain minimum precision. Those quality guarantees come at a cost in terms of optimizer performance: all existing algorithms for multi-objective query optimization have at least exponential time complexity in the number of tables (potentially higher depending on the number of Pareto plans). This means that they cannot be applied for queries with elevated number of tables.

For the traditional query optimization problem with one cost metric, there is a rich body of work proposing heuristics and randomized algorithms~\cite{Swami1988, Steinbrunn1997, ioannidis1990randomized, Bennett1991}. Those algorithms offer no formal quality guarantees on how far the generated plans are from the theoretical optimum but often generate good plans in practice. They have polynomial complexity in the number of tables and can be applied to much larger queries than exhaustive approaches. Up to date, corresponding approaches for multi-objective query optimization are missing entirely (and we will show later that algorithms for traditional query optimization perform poorly for the multi-objective case). In this paper, we close that gap and present the first randomized algorithm for multi-objective query optimization with polynomial time complexity. 

Existing algorithms for single- or multi-objective query optimization typically exploit only one out of two fundamental insights about the query optimization problem: dynamic programming based algorithms~\cite{Selinger1979, Kossmann2000, Trummer2014, Trummer2015a, Trummer2015} exploit its decomposability, i.e.\ the fact that a query optimization problem can be split into smaller sub-problems such that optimal solutions (query plans) for the sub-problems can be combined into optimal solutions for the original problem. Randomized algorithms such as iterative improvement, simulated annealing, or two-phase optimization exploit a certain near-convexity (also called well shape~\cite{Ioannidis1991}) of the standard cost functions when using suitable neighboring relationships in the query plan space. There is no reason why both insights shouldn't be exploited within the same algorithm and we do so: our algorithm improves plans using a multi-objective generalization of hill climbing, thereby exploiting near-convexity. It also maintains a plan cache storing partial Pareto-optimal plans generating potentially useful intermediate results. Newly generated plans are decomposed and dominated sub-plans are replaced by partial plans from the cache. Therefore we exploit decomposability as well.

Our algorithm is iterative and performs the following steps in each iteration. First, a query plan is randomly generated. Second, the plan is improved using local search until a local optimum is reached. Third, based on the locally optimal plan we restrict the plan space to plans that are similar in certain aspects (we provide details in the following paragraphs). We approximate the Pareto plan set within that restricted plan space. For that approximation, we might re-use partial plans that were generated in prior iterations if they realize a better cost tradeoff than the corresponding sub-plans of the locally optimal plan.

For the second step, we use a multi-objective version of hill climbing that  exploits several properties of the query optimization problem to reduce time complexity significantly compared to a naive version. First, we exploit the multi-objective principle of optimality for query optimization~\cite{ganguly1992query} stating that replacing a sub-plan by another sub-plan whose cost is not Pareto-optimal cannot improve the cost of the entire plan. This allows to quickly discard local mutations that will not improve the overall plan. Second, we exploit that query plans can be recursively decomposed into sub-plans for which local search can be applied independently. This allows to apply many beneficial mutations simultaneously in different parts of the query tree and therefore reduces the number of complete query plans that need to be generated on the path from the random plan to a local optimum. Those optimizations reduce the time complexity comparing with a naive version and we found them to be critical to achieve good optimizer performance. 

For the third step, we restrict the plan space to plans that generate a similar set of intermediate results as the locally optimal plan that results from the second step. We consider plans that use the same join order as the locally optimal plan but different operator combinations. Also, we consider the possibility to replace sub-plans by plans from a cache (those plans can use a different join order than the locally optimal plan). The cache stores non-dominated partial plans for each potentially useful intermediate result we encountered during optimization so far. Finding the full Pareto plan set even within the restricted plan space may lead to prohibitive computational cost (the number of Pareto plans within the restricted space may grow exponentially in the number of query tables). We therefore approximate the Pareto plan set by a subset of plans (whose size grows polynomially in the number of query tables) realizing representative cost tradeoffs. The precision of that approximation is slowly refined over the iterations. This enables our algorithm to quickly find a coarse-grained approximation of the full Pareto plan set for the given query; at the same time, as we refine precision, the approximation converges to the real Pareto set. 

The insights underlying the design of our algorithm are the following: on the one hand, we observe that the same join order can often realize many Pareto-optimal cost tradeoffs when using different operator configurations. This is why we approximate in the third step a Pareto frontier based on a restricted set of join orders. On the other hand, the full Pareto frontier cannot be covered using only one join order. This is why we generate new plans and join orders in each iteration. 

We analyze our randomized algorithm experimentally, using different cost metrics, query graph structures, and query sizes. We compare against dynamic programming based approximation schemes that were previously proposed for multi-objective query optimization. While approximation schemes are preferable for small queries, we show that only randomized algorithms can handle larger query sizes. We evaluate our algorithm with queries joining up to 100 tables considering an unconstrained bushy plan space. Even in case of one cost metric, dynamic programming based approaches do not scale to such search space sizes. We also compare our algorithm against other randomized algorithms: the non-dominated sort genetic algorithm 2 (NSGA-II)~\cite{Deb2002} is a very popular multi-objective optimization algorithm for the number of plan cost metrics that we consider in our experiments. Genetic algorithms have been very successful for traditional query optimization~\cite{Bennett1991} so a comparison against NSGA-II (using the combination and mutation operators proposed for traditional query optimization) seems interesting. We also compare our algorithm against other multi-objective generalizations of well known randomized algorithms for traditional query optimization such as iterative improvement, simulated annealing, and two-phase optimization~\cite{Steinbrunn1997}. Our randomized algorithm outperforms all competitors significantly, showing that the combination of local search with plan decomposition is powerful.

We analyze the time complexity of our algorithm and show that each iteration has expected polynomial complexity. Our analysis includes in particular a study of the expected path length from a random plan to the nearest local optimum. Based on a simple statistical model of plan cost distributions, we can show that the expected path length grows at most linearly in the number of query tables.

In summary, the original scientific contributions of this paper are the following:
\begin{itemize}
\item We present the first polynomial time algorithm for multi-objective query optimization: a randomized algorithm that exploits several properties of the query optimization problem.
\item We analyze that algorithm formally, showing that each iteration (resulting in at least one query plan) has polynomial complexity in the number of query tables.
\item We evaluate our algorithm experimentally against previously published approximation schemes for multi-objective query optimization and several randomized algorithms. We show that our algorithm outperforms the other algorithms over a wide range of scenarios.
\end{itemize}

The remainder of this paper is organized as follows. We give an overview of related work in Section~\ref{relatedSec}. In Section~\ref{modelSec}, we introduce the formal model used in pseudo-code and formal analysis. We introduce the first randomized algorithm for multi-objective query optimization in Section~\ref{algSec} and analyze its complexity in Section~\ref{analysisSec}. In Section~\ref{experimentsSec}, we experimentally evaluate our algorithm in comparison with several baselines.
\section{Related Work}
\label{relatedSec}

Most work in query optimization treats the single-objective case, meaning that query plans are compared according to only one cost metric (usually execution time)~\cite{Bennett1991, Bruno, Selinger1979, Steinbrunn1997, Swami1988,  Vance1996a}. This problem model is however often insufficient: in a cloud scenario, users might be able to reduce query execution time when willing to pay more money for renting additional resources from the cloud provider~\cite{kllapi2011schedule}. On systems that process multiple queries concurrently, the tradeoff between the amount of dedicated system resources and query execution time needs to be considered~\cite{Trummer2014}. In those and other scenarios, query optimization becomes a multi-objective optimization problem.

Query optimization algorithms that have been designed for the case of one cost metric cannot be applied to the multi-objective case. They return only one optimal plan while the goal in multi-objective query optimization is usually to find a set of Pareto-optimal plans~\cite{Papadimitriou2001, Trummer2015, Trummer2015a}. This allows in particular to let users choose their preferred cost tradeoff out of a visualization of the available tradeoffs~\cite{Trummer2015a}. We will use several variations of single-objective randomized query optimization algorithms as baselines for our experiments. Note that mapping multi-objective optimization into a single-objective optimization problem using a weighted sum over different cost metrics with varying weights will not yield the Pareto frontier but at most a subset of it (the convex hull). 

One exhaustive optimization algorithm and multiple approximation schemes have been proposed for multi-objective query optimization~\cite{ganguly1992query, Trummer2014, Trummer2015a, Trummer2015}. They are based on dynamic programming and have all exponential complexity in the number of query tables. This means that they do not scale to large query sizes. In traditional single-objective optimization, randomized algorithms and heuristics are used for query sizes that cannot be handled by exhaustive approaches. No equivalent is currently available for multi-objective query optimization. In this work, we close that gap and propose the first polynomial time heuristic for multi-objective query optimization.

One of the most popular randomized algorithms for single-objective query optimization is the genetic algorithm~\cite{Bennett1991}. Also, multi-objective genetic algorithm variants are very popular for multi-objective optimization in general~\cite{CoelloCoello2006}. It seems therefore natural to use the crossover and mutation operators that have been proposed for traditional query optimization within a multi-objective genetic algorithm variant. We implemented a version of the widely used non-dominated sort genetic algorithm II (NSGA-II)~\cite{Deb2002} for our experiments. 

We focus in this paper on query optimization in the traditional sense, i.e.\ the search space is the set of available join orders and the selection of scan and join operators. This distinguishes our work for instance from work on multi-objective optimization of workflows~\cite{kllapi2011schedule, Simitsis2010} which does not consider alternative join orders. Prior work by Papadimitriou and Yannakakis~\cite{Papadimitriou2001} also aims at optimizing operator selections for a fixed join order and addresses therefore a different problem than we do. We focus on generic multi-objective query optimization as defined in our prior work~\cite{Trummer2014} and our algorithm is not bound to specific combinations of cost metrics or scenarios such as precision-time~\cite{Agarwal2012} or energy-time tradeoffs~\cite{xu2012pet}.
\section{Formal Model}
\label{modelSec}

A query $q$ is modeled as a set of tables that need to be joined. This query model is simplistic but often used in query optimization~\cite{Swami1988, ganguly1992query, Trummer2014}; extending a query optimization algorithm that optimizes queries represented as table sets to more complex query models is standard~\cite{Selinger1979}. A query plan $p$ for a query $q$ specifies how the data described by the query can be generated by a series of scan and join operations. The plan describes the join order and the operator implementation used for each scan and join. We write \Call{ScanPlan}{$q,op$} to denote a plan scanning the single table $q$ ($|q|=1$) using scan operator $op$. We write \Call{JoinPlan}{$outer,inner,op$} to denote a join plan that joins the results produced by an outer plan $outer$ with the results produced by an inner plan $inner$ using join operator $op$. 

We denote by $p.rel$ the set of tables joined by a plan $p$. It is \Call{ScanPlan}{$q,op$}$.rel=q$ and \Call{JoinPlan}{$po,pi,op$}$.rel=po.rel\cup pi.rel$. For join plans we denote by $p.outer$ the outer plan and by $p.inner$ the inner plan. The property $p.isJoin$ yields \textbf{true} for join plans (where $|p.rel|>1$) and \textbf{false} for scan plans (where $|p.rel|=1$). 

We compare query plans according to their execution cost. We consider multiple cost metrics that might for instance include monetary fees (in a cloud scenario~\cite{kllapi2011schedule}), energy consumption~\cite{xu2012pet}, or various metrics of system resource consumption such as the number of used cores or the amount of consumed buffer space~\cite{Trummer2014} in addition to execution time. We consider cost metrics in the following, meaning that a lower value is preferable. It is straight forward to transform a quality metric on query plans such as result precision~\cite{Agarwal2012} into a cost metric (e.g., result precision can be transformed into the precision loss cost metric as shown in our prior work~\cite{Trummer2014}). Hence we study only cost metrics in the following without restriction of generality.

We denote by $p.cost\in\mathbb{R}^l$ the cost vector associated with a query plan. Each vector component represents cost according to a different cost metric out of $l$ cost metrics. We focus on optimization in this paper and assume that cost models for all considered cost metrics are available. The algorithms that we discuss in this paper are generic and can be applied to a broad set of plan cost metrics. 

In the special case of one cost metric, we say that plan $p_1$ is better than plan $p_2$ if the cost of $p_1$ is lower. Pareto-dominance is the generalization to multiple cost metrics. In case of multiple cost metrics, we say that plan $p_1$ dominates $p_2$, written $p_1\preceq p_2$ if $p_1$ has lower or equivalent cost to $p_2$ according to each considered cost metric. We say that $p_1$ strictly dominates $p_2$, written $p_1\prec p_2$, if $p_1\preceq p_2$ and $p_1.cost\neq p_2.cost$, meaning that $p_1$ has lower or equivalent cost than $p_2$ according to each metric and lower cost in at least one metric. We apply the same terms and notations to cost vectors in general, e.g.\ we write $c_1\preceq c_2$ for two cost vectors $c_1$ and $c_2$ to express that there is no cost metric for which $c_1$ contains a higher cost value than $c_2$.

Considering a set $P$ of alternative plans generating the same results, we call each plan $p\in P$ Pareto-optimal if there is no other plan $\widetilde{p}\in P$ that strictly dominates $p$. For a given query, the Pareto plan set is the set of plans for $q$ that are Pareto-optimal within the set of all possible query plans for $q$. The full Pareto plan set is often too large to be calculated in practice. This is why we rather aim at approximating the Pareto plan set. The following definitions are necessary to establish a measure of how well a given plan set approximates the real Pareto set.

A plan $p_1$ approximately dominates a plan $p_2$ with approximation factor $\alpha\geq1$, written $p_1\preceq_{\alpha}p_2$, if $p_1.cost\preceq\alpha \cdot p_2.cost$. This means that the cost of $p_1$ is not higher than the cost of $p_2$ by more than factor $\alpha$ according to each cost metric. Considering a set $P$ of plans, an $\alpha$-approximate Pareto plan set $P_{\alpha}\subseteq P$ is a subset of $P$ such that $\forall p\in P\exists\widetilde{p}\in P_{\alpha}:\widetilde{p}\preceq_{\alpha}p$, i.e.\ for each plan $p$ in the full set there is a plan in the subset that approximately dominates $p$. The (approximate) Pareto frontier are the cost vectors of the plans in the (approximate) Pareto set.

The goal of multi-objective query optimization is to find $\alpha$-approximate Pareto plan sets for a given input query $q$. We compare different incremental optimization algorithms in terms of the $\alpha$ values (i.e., how well the plan set generated by those algorithms approximates the real Pareto set) that they produce after certain amounts of optimization time.
\section{Algorithm Description}
\label{algSec}

We describe a randomized algorithm for multi-objective query optimization. Section~\ref{overviewSub} describes the main function, the following two subsections describe sub-functions. Section~\ref{climbingSub} describes a multi-objective hill climbing variant that executes multiple plan transformations in one step for maximal efficiency. Section~\ref{frontierSub} describes how we generate a local Pareto frontier approximation for a given join order, using non-dominated partial plans from a plan cache and trying out different operator configurations.

\subsection{Overview}
\label{overviewSub}

\begin{algorithm}[t!]
\renewcommand{\algorithmiccomment}[1]{// #1}
\begin{algorithmic}[1]
\State \Comment{Returns approximate Pareto plan set for query $q$}
\Function{RandomMOQO}{$q$}
\State \Comment{Initialize partial plan cache and iteration counter}
\State $P\gets\emptyset$
\State $i\gets1$
\State \Comment{Refine frontier approximation until timeout}
\While{No Timeout}
\State \Comment{Generate random bushy query plan}
\State $plan\gets$\Call{RandomPlan}{$q$}
\State \Comment{Improve plan via fast local search}
\State $optPlan\gets$\Call{ParetoClimb}{$plan$}
\State \Comment{Approximate Pareto frontier}
\State $P\gets$\Call{ApproximateFrontiers}{$optPlan,P,i$}
\State $i\gets i+1$
\EndWhile
\State \Return{$P[q]$}
\EndFunction
\end{algorithmic}
\caption{Main function.\label{mainAlg}}
\end{algorithm}

Algorithm~\ref{mainAlg} is the main function of our optimization algorithm. The input is a query $q$ and the output a set of query plans that approximate the Pareto plan set for $q$.

Our algorithm is iterative and refines the approximation of the Pareto frontier in each iteration. Each iteration consists of three principal steps: a random query plan is generated (local variable $plan$ in the pseudo-code), it is improved via a multi-objective version of hill climbing, and afterwards the improved plan (local variable $optPlan$) is used as base to generate a local Pareto frontier approximation. For the latter step, a plan cache (local variable $P$ in the pseudo-code) is used that stores for each intermediate result (i.e., a subset $s\subseteq q$ of joined tables) that we encountered so far a set of non-dominated partial plans. For the locally optimal plan resulting from hill climbing, we consider plans that can be obtained by varying the operator configurations but not the join order. In addition, we consider replacing sub-plans by non-dominated partial plans from the plan cache. 

Within that restricted plan space, we do not search for the entire Pareto frontier as its size can be exponential in the number of query tables. Instead, we search for an approximation that has guaranteed polynomial size in the number of query tables. The precision of those approximations is refined with increasing number of iterations: our goal is to obtain a coarse-grained approximation of the entire Pareto frontier quickly so we start with a coarse approximation precision to quickly explore a large number of join orders. In later iterations, the precision is refined to allow to better exploit the set of join orders that was discovered so far. As the frontier approximation precision depends on the iteration number, function~\textproc{ApproximateFrontiers} obtains the iteration counter $i$ as input parameter in addition to the locally optimal plan and the plan cache. All non-dominated partial plans generated during the frontier approximation are inserted into the plan cache $P$ and might be reused in following iterations. 

After the timeout, the result plan set is contained in the plan cache and associated with table set $q$, the entire query table set (we use the array notation $P[q]$ to denote the set of cached Pareto plans that join table set $q$). Note that we can easily use different termination conditions than a timeout. In particular, in case of interactive query optimization where users choose an execution plan based on a visualization of available cost tradeoffs~\cite{Trummer2015a}, optimization ends once the user selects a query plan for execution from the set of plans generated so far.

Our algorithm exploits two ideas that have been very successful in traditional query optimization but are typically used in separate algorithms: our algorithm exploits a certain near-convexity of typical plan cost functions~\cite{Ioannidis1991} by using local search (function~\textproc{ParetoClimb}) to improve query plans. It exploits however at the same time that the query optimization problem can be decomposed into sub-problems, meaning that Pareto-optimal plans joining table subsets can be combined into Pareto-optimal plans joining larger table sets. It is based on the insight that the same join order often allows to construct multiple Pareto-optimal cost tradeoffs by varying operator implementations but takes into account at the same time that not all optimal cost tradeoffs can be found considering only one join order. We describe the two principal sub-functions of Algorithm~\ref{mainAlg}, function~\textproc{ParetoClimb} and function~\textproc{ApproximateFrontiers}, in more detail in the following subsections.

Note finally that the algorithm can easily be adapted to consider different join order spaces (e.g., left-deep plans) by exchanging the random plan generation method and the set of considered local transformations.

\subsection{Pareto Climbing}
\label{climbingSub}

\begin{algorithm}[t!]
\renewcommand{\algorithmiccomment}[1]{// #1}
\begin{algorithmic}[1]
\State \Comment{Plan $p_1$ is better than $p_2$ if it produces the same}
\State \Comment{data format as $p_2$ and has dominant cost.}
\Function{Better}{$p_1,p_2$}
\State \Return{\Call{SameOutput}{$p_1,p_2$}$\wedge(p_1\prec p_2)$}
\EndFunction
\vspace{0.1cm}
\State \Comment{Keeps one Pareto plan per output format.}
\Function{Prune}{$plans,newPlan$}
\If{$\nexists p\in plans:$\Call{Better}{$p,newPlan$}}
\State $plans\gets \{p\in plans|\neg$\Call{Better}{$newPlan,p$}$\}$
\State $plans\gets plans\cup\{newPlan\}$
\EndIf
\State \Return{$plans$}
\EndFunction
\vspace{0.1cm}
\State \Comment{Improve plan $p$ by parallel local transformations}
\Function{ParetoStep}{$p$}
\State \Comment{Initialize optimal mutations of this plan}
\State $pPareto\gets\emptyset$
\If{$p.isJoin$}
\State \Comment{Improve sub-plans by recursive calls}
\State $outerPareto\gets$\Call{ParetoStep}{$p.outer$}
\State $innerPareto\gets$\Call{ParetoStep}{$p.inner$}
\State \Comment{Iterate over all improved sub-plan pairs}
\For{$outer\in outerPareto$}
\For{$inner\in innerPareto$}
\State $p.outer\gets outer$
\State $p.inner\gets inner$
\State \Comment{Mutations for specific sub-plan pair}
\For{$mutated\in$\Call{Mutations}{$p$}}
\State $pPareto\gets$\Call{Prune}{$pPareto, mutated$}
\EndFor
\EndFor
\EndFor
\Else
\State \Comment{$p$ is single-table scan}
\For{$mutated\in$\Call{Mutations}{$p$}}
\State $pPareto\gets$\Call{Prune}{$pPareto, mutated$}
\EndFor
\EndIf
\State \Return{$pPareto$}
\EndFunction
\vspace{0.1cm}
\State \Comment{Climbs until plan $p$ cannot be improved further.}
\Function{ParetoClimb}{$p$}
\State $improving\gets$\textbf{true}
\While{$improving$}
\State $improving\gets$\textbf{false}
\State $mutations\gets$\Call{ParetoStep}{$p$}
\If{$pm\in mutations:pm\prec p$}
\State $p\gets pm$
\State $improving\gets$\textbf{true}
\EndIf
\EndWhile
\State \Return{$p$}
\EndFunction
\end{algorithmic}
\caption{Fast multi-objective hill climbing performing mutations in independent plan sub-trees simultaneously.\label{ParetoClimbAlg}}
\end{algorithm}

Algorithm~\ref{ParetoClimbAlg} shows the pseudo-code of function~\textproc{ParetoClimb} (and of several auxiliary functions) that is used by Algorithm~\ref{mainAlg} to improve query plans via local search. The input is a query plan $p$ to improve and the output is a locally optimal query plan that was reached from the input plan using a series of local transformations.

Hill climbing was already used in traditional query optimization~\cite{Swami1988} but our hill climbing variant differs from the traditional version in several aspects that are discussed in the following. A first obvious difference is that we consider multiple cost metrics on query plans while traditional query optimization considers only execution time. In case of one cost metric, hill climbing moves from one plan to a neighbor plan (i.e., a plan that can be reached via a local transformation~\cite{Steinbrunn1997}) if the neighbor plan has lower cost than the first one. Our multi-objective version moves from a first plan to a neighbor if the neighbor strictly dominates the first one in the Pareto sense, i.e.\ the second plan has lower or equivalent cost according to all metrics and lower cost according to at least one. 

In principle, there can be multiple neighbors that strictly dominate the start plan while different neighbors do not necessarily dominate each other. In such cases, it cannot be determined which neighbor is the best one to move to and all neighbors might be required for the full Pareto frontier. In order to avoid a combinatorial explosion, we still chose to arbitrarily select one neighbor that strictly dominates the start plan instead of opening different path branches during the climb. The goal of function~\textproc{ParetoClimb} within Algorithm~\ref{mainAlg} is to find one good plan while function~\textproc{ApproximateFrontiers} (which is discussed in the next subsection) will take care of exploring alternative cost tradeoffs.

Unlike most prior hill climbing variants used for traditional query optimization~\cite{Swami1988, Steinbrunn1997}, we chose to exhaustively explore all neighbor plans in each step of the climb instead of randomly sampling a subset of neighbors. We initially experimented with random sampling of neighbor plans which led however to poor performance. We believe that this is due to the fact that dominating neighbors become more and more sparse as the number of considered cost metrics grows. Using the simple statistical model that we introduce in Section~\ref{analysisSec}, the probability of finding a dominating neighbor decreases exponentially in the number of cost metrics. Furthermore, sampling introduces overhead and makes it harder if not impossible to use the techniques for complexity reduction that we describe next.

Reducing the time complexity of local search as much as possible is crucial as function~\textproc{ParetoClimb} is called in each iteration of Algorithm~\ref{mainAlg}. We exploit properties of the multi-objective query optimization problem in order to make our implementation of local search much more efficient than a naive implementation. A naive hill climbing algorithm iterates the following steps until a local optimum is reached: in each step, it traverses all nodes of the current query plan tree and applies to each node a fixed set of local mutations. For each mutation and plan node, a new complete neighbor plan is created and its cost is calculated. Based on that cost, the next plan is selected among the neighbors. This naive approach has per step quadratic complexity in the number of plan nodes (which is linear in the number of query tables). 

For a first improvement, we can exploit the principle of optimality for multi-objective query optimization~\cite{ganguly1992query}. After applying a local transformation to one specific node in the query tree, it is not always necessary to calculate the cost of the completed plan (at the tree root) in order to determine whether that mutation reduces the plan cost. Due to the principle of optimality, improving a sub-plan cannot worsen the entire plan (we assume for the current discussion that all alternative plans produce the same data representation, neglecting for instance the impact of interesting orders~\cite{Selinger1979}). In many cases, reducing the cost of a sub-plan even guarantees that the cost of the overall plan decreases as well\footnote{This is guaranteed for cost metrics such as energy consumption, monetary cost, precision loss, and execution time when considering plans without parallel branches~\cite{Trummer2014} where the cost of a plan is calculated as weighted sum or product over the cost of its sub-plans. It might not hold in special cases (e.g., when replacing a sub-plan that is not on the critical path in a parallel execution scenario by a faster one).}. This means that we can at least exclude that a certain mutation reduces the cost of the entire plan if it worsens the cost of the sub-plan to which it was applied. As the cost of the sub-plan can be recalculated in constant time (when treating the number of cost metrics as a constant as it is common in multi-objective query optimization~\cite{Trummer2014}), this simple optimization already reduces the complexity per step from quadratic in the number of tables to linear whenever the chances that a local cost improvement does not yield a global improvement are negligible.

While the last optimization reduces the complexity per climbing step, the optimization discussed next tends to reduce the number of steps required to reach the nearest local optimum. Query plans are represented as trees and we can simultaneously apply mutations in independent sub-trees; it is not necessary to generate complete query trees after each single mutation which is what the naive hill climbing variant does. If we reduce the cost of several sub-trees simultaneously then the cost of the entire plan cannot worsen either due to the principle of optimality. Applying multiple beneficial transformations in different parts of the query tree simultaneously reduces the number of completed query trees that are created on the way to the local optimum. 

Note that all discussed optimizations are also useful to improve the efficiency of local search for traditional query optimization with one cost metric. Local search has already been used for traditional query optimization but we were not able to find any discussion of the aforementioned issues in the literature while they have significant impact on the performance (the second optimization alone reduced the average time for reaching local optima from randomly selected plans by over one order of magnitude for queries with 50 tables in a preliminary benchmark).

Algorithm~\ref{ParetoClimbAlg} integrates all of the aforementioned optimizations. Function~\textproc{ParetoClimb} performs plan transformations until the plan cannot be improved anymore, i.e.\ there is no neighbor plan with dominant cost. Function~\textproc{ParetoStep} realizes one transformation step. It may return multiple Pareto-optimal plan mutations that produce data in different representations (e.g., materialized versus non-materialized). We must do so since sub-plans producing different data representations cannot be compared as the data representation can influence the cost (or applicability) of other operations higher-up in the plan tree. We assume that the standard mutations for bushy query plans~\cite{Steinbrunn1997} are considered for each node in the plan tree. Function~\textproc{ParetoStep} might however mutate multiple nodes in the tree during one call: when treating join plans then the outer and the inner sub-plan are both replaced by dominant mutations via a recursive call. We try out each combination of potentially improved sub-plans and try all local transformations for each combination. The resulting plans are pruned such that only one non-dominated plan is kept for each possible output data representation. Function~\textproc{Better} is used during pruning to compare query plans and returns \textbf{true} if and only if a first plan produces the same output data representation as a second (tested using function~\textproc{SameOutput}) and the cost vector of the first plan strictly dominates the one of the second. 

The following example illustrates how Algorithm~\ref{ParetoClimbAlg} works.

\begin{example}
We assume for simplicity that only one plan cost metric is considered, that we have only one scan and join operator implementation such that only join order matters, and that the only mutation is the exchange of outer and inner join operands. We invoke function~\textproc{ParetoClimb} on an initial query plan $(S\Join T)\Join(R\Join U)$. Function~\textproc{ParetoClimb} invokes function~\textproc{ParetoStep} to improve the initial plan. The root of the initial query plan is a join between $(S\Join T)$ and $(R\Join U)$. Function~\textproc{ParetoStep} tries to improve the two operands of that join via recursive calls before considering mutations at the plan root. Hence one instance of \textproc{ParetoStep} is invoked to improve the outer operand $(S\Join T)$ and another instance is invoked to improve the inner operand $(R\Join U)$. The instance of function \textproc{ParetoStep} treating $(S\Join T)$ spawns new instances for improving the two join operands $S$ and $T$. As we consider no scan operator mutations here, those invocations do not have any effect. After trying to improve $S$ and $T$, the instance treating $(S\Join T)$ tries the mutation $(T\Join S)$. Assume that this mutation reduces cost. Then the improved sub-plan $(T\Join S)$ will be returned to the instance of function \textproc{ParetoStep} treating the initial plan. Assume that the sub-plan $(R\Join U)$ cannot be improved by exchanging outer and inner join operand. Then the top-level instance of function \textproc{ParetoStep} will try the mutation $(R\Join U)\Join (T\Join S)$ and compare its execution cost to the cost of $(T\Join S)\Join(R\Join U)$. The cheaper plan is returned to function~\textproc{ParetoClimb} which detects that the initial plan has been improved. This function performs another iteration of the main loop as changing the initial plan can in principle enable new transformations that yield further improvements. This is not possible in the restricted setting of our example and hence function~\textproc{ParetoClimb} terminates after two iterations.
\end{example}

\subsection{Frontier Approximation}
\label{frontierSub}

\begin{algorithm}[t!]
\renewcommand{\algorithmiccomment}[1]{// #1}
\begin{algorithmic}[1]
\State \Comment{Checks if plan $p_1$ is significantly better than $p_2$}
\State \Comment{using coarsening factor $\alpha$ for cost comparison.}
\Function{SigBetter}{$p_1,p_2,\alpha$}
\State \Return{\Call{SameOutput}{$p_1,p_2$}$\wedge p_1\preceq_{\alpha}p_2$}
\EndFunction
\vspace{0.1cm}
\State \Comment{Returns an $\alpha$-approximate Pareto frontier.}
\Function{Prune}{$plans,newP,\alpha$}
\State \Comment{Can we approximate the cost of the new plan?}
\If{$\nexists p\in plans:$\Call{SigBetter}{$p,newP,\alpha$}}
\State $plans\gets \{p\in plans|\neg$\Call{SigBetter}{$newP,p,1$}$\}$
\State $plans\gets plans\cup\{newP\}$
\EndIf
\State \Return{$plans$}
\EndFunction
\vspace{0.1cm}
\State \Comment{Approximates the Pareto frontier for each}
\State \Comment{intermediate result that appears in plan $p$,}
\State \Comment{using partial plans from the plan cache $P$. }
\State \Comment{The precision depends on the iteration count $i$.}
\Function{ApproximateFrontiers}{$p,P,i$}
\State \Comment{Calculate target approximation precision}
\State $\alpha\gets 25\cdot 0.99^{\lfloor i/25\rfloor}$
\If{$p.isJoin$}
\State \Comment{Approximate outer and inner plan frontiers}
\State $P\gets$\Call{ApproximateFrontiers}{$p.outer,P,i$}
\State $P\gets$\Call{ApproximateFrontiers}{$p.inner,P,i$}
\State \Comment{Iterate over outer Pareto plans}
\For{$outer\gets P[p.outer.rel]$}
\State \Comment{Iterate over inner Pareto plans}
\For{$inner\gets P[p.inner.rel]$}
\State \Comment{Iterate over applicable join operators}
\For{$op\in$\Call{JoinOps}{$outer,inner$}}
\State \Comment{Generate new plan and prune}
\State $np\gets$\Call{JoinPlan}{$outer,inner,op$}
\State $P[p.rel]\gets$\Call{Prune}{$P[p.rel],np,\alpha$}
\EndFor
\EndFor
\EndFor
\Else
\State \Comment{Iterate over applicable scan operators}
\For{$op\in$\Call{ScanOps}{$p.rel$}}
\State $np\gets$\Call{ScanPlan}{$p.rel,op$}
\State $P[p.rel]\gets$\Call{Prune}{$P[p.rel],np,\alpha$}
\EndFor
\EndIf
\State \Comment{Return updated plan cache}
\State \Return{$P$}
\EndFunction
\end{algorithmic}
\caption{Approximating the Pareto frontiers of all intermediate results occuring within given plan.\label{frontierAlg}}
\end{algorithm}

The goal of Pareto climbing is to find one plan that is at least locally Pareto-optimal. Having such a plan, it is often possible to obtain alternative optimal cost tradeoffs by varying the operator implementations while reusing the same join order\footnote{More precisely, this is possible when choosing an appropriate formalization of the plan space: when considering tradeoffs between buffer space consumption and execution time, we can for instance introduce different versions of the standard join operators that work with different amounts of buffer space. When considering tradeoffs between result precision and execution time in approximate query processing, we might introduce different scan operator versions associated with different sample densities. In a cloud scenario, we can introduce operator versions that are associated with different degrees of parallelism, allowing to trade monetary cost for execution time.}. We exploit that fact in function~\textproc{ApproximateFrontiers} whose pseudo-code is shown in Algorithm~\ref{frontierAlg}. 

Function~\textproc{ApproximateFrontiers} obtains a query plan $p$ whose join order it exploits, the plan cache $P$ mapping intermediate results to non-dominated plans 
generating them, and the iteration counter $i$ (counting iterations of the main loop in Algorithm~\ref{mainAlg}) as input. The output is an updated plan cache in which the non-dominated plans generated in the current invocation have been inserted for the corresponding intermediate results.

Function~\textproc{ApproximateFrontiers} starts by choosing the approximation factor $\alpha$ which depends on the iteration number $i$. A higher approximation factor reduces the time required for approximation but a lower approximation factor yields a more fine-grained approximation. We will see in Section~\ref{analysisSec} that \textproc{ApproximateFrontiers} has dominant time complexity within each iteration so a careful choice of the approximation factor is crucial. Multi-objective query optimization can be an interactive process in which alternative cost tradeoffs are visualized to the user such that he can make his choice~\cite{Trummer2015a}. Especially in such scenarios, it is beneficial to rather obtain a coarse-grained approximation of the entire Pareto frontier quickly than to obtain a very fine-grained approximation of only a small part of it. As approximating the entire Pareto frontier requires in general considering many join orders (see our remark at the end of this subsection), we do not want to spend too much time at the beginning exploiting one single join order. For that reason we start with a coarse-grained approximation factor that is reduced as iterations progress. We have tried different formulas for choosing $\alpha$ and the formula given in the pseudo-code worked best for a broad range of scenarios.


Having chosen the approximation precision, the function approximates the Pareto frontier for scan plans by trying out all available scan operators on the given table. For join plans, a frontier is first generated for the outer and inner operand using recursive calls (so we traverse the query plan tree in post-order). After that, we consider each possible pair of a plan from the outer frontier with a plan from the inner frontier and each applicable join operator and generate one new plan for each such combination. Newly generated plans are pruned again but the definition of the pruning function differs from the one we used in Algorithm~\ref{ParetoClimbAlg}. For the same output data properties, the pruning function in Algorithm~\ref{frontierAlg} might now keep multiple plans that realize different optimal cost tradeoffs. However, in contrast to the previous pruning variant, new plans are only inserted if their cost cannot be approximated by any of the plans already in the pruned plan set. We will see in Section~\ref{analysisSec} that this pruning function guarantees that the number of plans stored for a single table set is bounded by a polynomial in the number of query tables. 

Note that function~\textproc{ApproximateFrontiers} does not only try different operator combinations for a given join order. Instead, we consider all non-dominated plans that are cached for the intermediate results generated by the input plan $p$. The plans we consider might have been inserted into the plan cache in prior iterations of the main loop and might therefore use different join orders. The plan cache is our mean of sharing information across different iterations of the main loop. As iterations continue, the content of the plan cache will more and more resemble the set of partial plans that is generated by dynamic programming based approximation schemes proposed for multi-objective query optimization~\cite{Trummer2014, Trummer2015a}. However, instead of approximating the frontier for each possible intermediate result (implying exponential complexity), we only treat table sets that are used by locally Pareto-optimal plans.

Note finally that it is in general not possible to obtain all Pareto optimal query plans by varying operators for a fixed join order. Considering again the tradeoff between buffer space and execution time, a left-deep plan using pipelining might for instance minimize execution time while a non-pipelined bushy plan achieves the lowest buffer footprint. Or in case of approximate query processing, different tradeoffs between result precision and execution time can be achieved by varying the sample size generated by scan operators so the output size for each table and hence the optimal join order depends on the operator configurations. This means that we need to consider different join orders in order to obtain the full Pareto set; we cannot decompose query optimization into join order selection followed by operator selection. Our algorithm respects that fact.

\section{Complexity Analysis}
\label{analysisSec}

We analyze the time and space complexity of the algorithm presented in Section~\ref{algSec}. We call that algorithm short RMQ for randomized multi-objective query optimizer.

We denote by $n$ the number of query tables to be joined. The number of cost metrics according to which query plans are compared is specified by $l$. Similar to prior work~\cite{Trummer2014, Trummer2015a}, we treat $n$ as a variable and $l$ as a constant: users choose the query size but the set of relevant cost metrics is usually determined by the query processing context. As in prior work~\cite{ganguly1992query}, we simplify the following formulas by assuming only one join operator while the generalization is straight-forward. We also neglect interesting orders such that query plans joining the same tables are only compared according to their cost vectors.

We analyze the time complexity of one iteration. Each iteration consists of three steps (random plan generation, local search, and frontier approximation); we analyze the complexity of each of those steps in the following. Random plan sampling (function~\textproc{RandomPlan} in Algorithm~\ref{mainAlg}) can be implemented with linear complexity as shown by the following lemma.

\begin{lemma}
Function~\textproc{RandomPlan} executes in $O(n)$ time.\label{randomComplexityLemma}
\end{lemma}
\begin{proof}
Query plans are labeled binary trees whose leaf nodes represent input tables. Quiroz shows how to generate binary trees in $O(n)$~\cite{Quiroz1989}. Labels representing input tables and operators can be selected in $O(n)$ as well. Estimating the cost of a query plan according to all cost metrics is in $O(ln)=O(n)$ time (as $l$ is a constant). 
\end{proof}

Next we analyze the complexity of local search, realized by function~\textproc{ParetoClimb} in Algorithm~\ref{mainAlg}. We first analyze the complexity of function~\textproc{ParetoStep} realizing a single step on the path to the next local optimum.

\begin{lemma}
Function~\textproc{ParetoStep} executes in $O(n)$ time.\label{neighborsComplexityLemma}
\end{lemma}
\begin{proof}
We apply a constant number of mutations at each of the $O(n)$ plan nodes. We assume that plans are only pruned based on their cost values such that each instance of function~\textproc{ParetoStep} returns only one non-dominated plan. Plan comparisons take $O(l)=O(1)$ time which leads to the postulated complexity. 
\end{proof}

The expected time complexity of local search depends of course on the number of steps required to reach the next local Pareto optimum (a plan that is not dominated by any neighbor). We analyze the expected path length based on a simple statistical model in the following: we model the cost of a random plan according to one cost metric as random variable and assume that the random cost variables associated with different metrics are independent from each other. This assumption is simplifying but standard in the analysis of multi-objective query optimization algorithms~\cite{ganguly1992query, Trummer2015}.

\begin{lemma}
The probability that a randomly selected plan dominates another is $(1/2)^l$.\label{domProbLemma}
\end{lemma}
\begin{proof}
The probability that a randomly selected plan dominates another plan according to one single cost metric is $1/2$. Assuming independence between different cost metrics, the probability that a random plan dominates another one according to all $l$ cost metrics is $(1/2)^l$.
\end{proof}

\begin{lemma}
The probability that none of $n$ plans dominates all plans in a plan set of cardinality $i$ is $(1-(1/2)^{li})^n$.
\end{lemma}
\begin{proof}
The probability that one plan dominates another is $(1/2)^l$ according to Lemma~\ref{domProbLemma}. The probability that one plan dominates all of $i$ other plans is $(1/2)^{li}$, assuming independence between the dominance probabilities between different plan pairs. The probability that a plan does not dominate all of $i$ other plans is $1-(1/2)^{li}$. The probability that none of $n$ plans dominates all of $i$ plans is $(1-(1/2)^{li})^n$.
\end{proof}

We denote by $u(n,i)=(1-(1/2)^{li})^n$ the probability that none of $n$ plans dominates all $i$ plans. We simplify in the following assuming that each plan has exactly $n$ neighbors.

\begin{theorem}
The expected number of plans visited by the hill climbing algorithm until finding a local Pareto optimum is $\sum_{i=1..\infty}i\cdot u(n,i)\cdot \prod_{j=1..i-1}(1-u(n,j))$.\label{nrNodesTh}
\end{theorem}
\begin{proof}
Our hill climbing algorithm visits a sequence of plans such that each plan is a neighbor of its predecessor and dominates its predecessor. Pareto dominance is a transitive relation and hence, as each plan dominates its immediate predecessor, each plan dominates all its predecessors. Then the probability of one additional step corresponds to the probability that at least one of the neighbors of the current plan dominates all plans encountered on the path so far. The probability that a local optimum is reached after $i$ plan nodes is the probability that none of the $n$ neighbors of the $i$-th plan dominates all $i$ plans on the path which is $u(n,i)$. The a-priori probability for visiting $i$ plans in total is then $u(n,i)\cdot \prod_{j=1..i-1}(1-u(n,j))$. The expected number of visited plans is $\sum_{i=1..\infty}i\cdot u(n,i)\cdot \prod_{j=1..i-1}(1-u(n,j))$.
\end{proof}

We bound the formula given by Theorem~\ref{nrNodesTh}.

\begin{theorem}
The expected path length from a random plan to the next local Pareto optimum is in $O(n)$.\label{pathLengthTheorem}
\end{theorem}
\begin{proof}
Assume that the hill climbing algorithm encounters at least $n$ plans on its path to the local Pareto optimum. We can write the expected number of additional steps as $\sum_{i=0..\infty}i\cdot u(n,n+i)\cdot \prod_{j=0..i-1}(1-u(n,j+n))$. Note that each additional step requires that one of the neighbors dominates at least $n$ plans on the path. Since $u(n,j)\leq1$, we can upper-bound the additional steps by $\sum_{i=0..\infty}i\cdot \prod_{j=0..i-1}(1-u(n,j+n))$. Since $(1-u(n,j))$ is anti-monotone in $j$, we can upper-bound that expression by $\sum_{i=0..\infty}i\cdot (1-u(n,n))^{i}$. This expression contains the first derivative of the infinite geometric series and we obtain $(1-u(n,n))/(1-(1-u(n,n)))^2\leq 1/(1-(1-u(n,n)))^2=1/u(n,n)^2$ for the additional steps. We study how quickly that expression grows in $n$. It is $1/u(n,n)^2=1/(1-(1/2)^{ln})^{2n}\leq 1/(1-(1/2)^{n})^{2n}\in O((1/(1-1/n)^n)^2)=O((1/e)^2)=O(1)$. The expected number of additional steps after $n$ steps is a constant.
\end{proof}

We finally calculated the expected time complexity of local search. Note that we make the pessimistic assumption that only one mutation is applied per path step while our algorithm allows in fact to apply many transformations together in one step.

\begin{theorem}
Function~\textproc{ParetoClimb} has expected time complexity in $O(n^2)$.
\end{theorem}
\begin{proof}
We combine the expected path length (see Theorem~\ref{pathLengthTheorem}) with the complexity per step (see Lemma~\ref{neighborsComplexityLemma}).
\end{proof}

At this point it might be interesting to compare the overhead of hill climbing, calculated before, with the benefit it provides by reducing the search space. The factor by which the search space size is reduced when focusing on local optima is given by the following lemma.

\begin{lemma}
The probability that a randomly selected plan is a local Pareto optimum is in $O((1-(1/2)^l)^{n})$. \label{probLocalOptLemma}
\end{lemma}
\begin{proof}
A plan is a local Pareto optimum if none of its neighbors dominates the plan. The neighbor plans are created by applying a constant number of mutations at each plan node. For a plan joining $n$ tables, the number of neighbors is therefore in $O(n)$. We simplify by assuming that the probability that a plan is dominated by one of its neighbors corresponds to the probability that it is dominated by a random plan which is $(1/2)^l$ according to Lemma~\ref{domProbLemma}. The probability that a plan is not dominated by one of its neighbors is $1-(1/2)^l$ and the probability that it is not dominated by any of its neighbors is in $O((1-(1/2)^l)^{n})$.
\end{proof}

Multi-objective hill climbing leads to an exponential reduction of search space size while the expected path length to the next local Pareto optimum grows only linearly in the number of query tables. Next we analyze the complexity of generating a frontier approximation. We denote by $\alpha$ the precision factor that is used in the analyzed iteration. The next lemma is based on a proof from prior work~\cite{Trummer2014} (Lemma~2 in the previous paper) that bounds the number of query plans such that no plan approximately dominates another. Using the same notation, we denote by $m$ the cardinality of the largest base table in the current database.

\begin{lemma}
The plan cache associates $O((n\log_{\alpha}m)^{l-1})$ \\plans with a table set of cardinality $n$.
\end{lemma}
\begin{proof}
New plans are only inserted into the plan cache if they are not approximately dominated by other plans joining the same tables, using approximation factor $\alpha$. The bound of $O((n\log_{\alpha}m)^{l-1})$ on the number of plans joining $n$ tables such that no plan approximately dominates another~\cite{Trummer2014} applies therefore to the plan cache.
\end{proof}

We denote by $b(n)$ the asymptotic bound on the number of plans stored per table set. Note that $b(n)$ grows monotonically in $n$.

\begin{theorem}
Function~\textproc{ApproximateFrontier} has time complexity $O(n\cdot b(n)^3)$.
\end{theorem}
\begin{proof}
The function treats each plan node in bottom-up order. For each node a set of new plans is generated and pruned by comparing it with alternative plans from the cache joining the same tables. We retrieve plans joining at most $n$ tables such that the number of plans retrieved for each table set is bounded by $b(n)$. Comparing one new plan against $b(n)$ stored plans is in $O(b(n))$. The complexity for treating inner nodes of the query plan tree (representing joins) is higher than the complexity of treating leaf nodes (representing scans). We generate $O(b(n)^2)$ new plans for an inner node which yields a per-node complexity of $O(b(n)^3)$ when taking pruning into account. Summing over all query plan nodes leads to the postulated complexity.
\end{proof}

The operation with dominant time complexity is the generation of the approximate frontier. This justifies that we start with a coarse-grained precision factor in order to explore a sufficient number of join orders quickly. The per-iteration complexity of RMQ follows immediately.

\begin{corollary}
RMQ has time complexity $O(n\cdot b(n)^3)$ per iteration.
\end{corollary}

We finally analyze the space complexity of RMQ. We analyze the accumulated space consumed after $i$ iterations. We assume that $b(n)$ designates the bound on the number of plans cached per table set for the precision factor $\alpha$ that is reached after $i$ iterations. 

\begin{theorem}
RMQ has space complexity $O(i\cdot n\cdot b(n))$.
\end{theorem}
\begin{proof}
Each plan generates $O(n)$ intermediate results. We approximate the Pareto frontiers of all intermediate results that are used by one locally optimal plan in each iteration. Each iteration therefore adds at most $O(n)$ plan sets to the plan cache. The number of plans cached for each intermediate result is bounded by $b(n)$ and each plan requires $O(1)$ space (as its sub-plans are already stored). 
\end{proof}

Considering multiple operator implementations changes time but not space complexity. We denote the number of implementations per operator by $r$. The asymptotic number of neighbor plans multiplies by $r$ and so does the expected path length to the nearest local optimum (this can be seen by substituting $n$ by $r\cdot n$ in Theorems~\ref{nrNodesTh} and \ref{pathLengthTheorem}). Hence the time complexity of local search multiplies by $r^2$. The time complexity of the frontier approximation multiplies by $r$ as we iterate over the operators in the innermost loop. Note that the number of plan cost metrics is traditionally treated as constant in multi-objective query optimization~\cite{Trummer2014, Trummer2015a}. 

\def\addAllPlots#1#2#3#4{
\addplot+[unbounded coords=jump, mark=Mercedes star, draw=black] table[header=true,col sep=tab, x index=0, y index=1] {plotsdata/newRun/scaled/#1M#2T#3#4.txt};
\addplot+[unbounded coords=jump, mark=asterisk, draw=black] table[header=true,col sep=tab, x index=0, y index=2] {plotsdata/newRun/scaled/#1M#2T#3#4.txt};
\addplot+[unbounded coords=jump, mark=x, draw=black] table[header=true,col sep=tab, x index=0, y index=3] {plotsdata/newRun/scaled/#1M#2T#3#4.txt};
\addplot+[unbounded coords=jump, mark=o, mark size=2, draw=blue] table[header=true,col sep=tab, x index=0, y index=4] {plotsdata/newRun/scaled/#1M#2T#3#4.txt};
\addplot+[unbounded coords=jump, mark=diamond*, mark size=1, draw=blue] table[header=true,col sep=tab, x index=0, y index=6] {plotsdata/newRun/scaled/#1M#2T#3#4.txt};
\addplot+[unbounded coords=jump, mark=triangle, mark size=2, draw=brown] table[header=true,col sep=tab, x index=0, y index=7] {plotsdata/newRun/scaled/#1M#2T#3#4.txt};
\addplot+[unbounded coords=jump, mark=square, mark size=1, draw=blue] table[header=true,col sep=tab, x index=0, y index=8] {plotsdata/newRun/scaled/#1M#2T#3#4.txt};
\addplot+[unbounded coords=jump, mark=x, mark size=2, draw=red] table[header=true,col sep=tab, x index=0, y index=9] {plotsdata/newRun/scaled/#1M#2T#3#4.txt};
}

\def\plotTitle#1#2#3{
\node at (group c#1r#2.north) [yshift=0.15cm] {#3};
}

\def\performancePlot#1#2{
\begin{groupplot}[group style={group size=3 by 5, x descriptions at=edge bottom, horizontal sep=1.5cm, vertical sep=0.4cm}, 
width=6cm, height=3.75cm,
xlabel=Time (s), ylabel=$\alpha$ (Log axis), xlabel near ticks, ylabel near ticks,
ymode=log, xmajorgrids, xmin=0, xmax=2.9, ymin=1,
log basis y=100, ymajorgrids]
\nextgroupplot
\addAllPlots{#1}{10}{chain}{#2}
\nextgroupplot
\addAllPlots{#1}{10}{cycle}{#2}
\nextgroupplot
\addAllPlots{#1}{10}{star}{#2}
\nextgroupplot
\addAllPlots{#1}{25}{chain}{#2}
\nextgroupplot
\addAllPlots{#1}{25}{cycle}{#2}
\nextgroupplot
\addAllPlots{#1}{25}{star}{#2}
\nextgroupplot
\addAllPlots{#1}{50}{chain}{#2}
\nextgroupplot
\addAllPlots{#1}{50}{cycle}{#2}
\nextgroupplot
\addAllPlots{#1}{50}{star}{#2}
\nextgroupplot
\addAllPlots{#1}{75}{chain}{#2}
\nextgroupplot
\addAllPlots{#1}{75}{cycle}{#2}
\nextgroupplot
\addAllPlots{#1}{75}{star}{#2}
\nextgroupplot
\addAllPlots{#1}{100}{chain}{#2}
\nextgroupplot
\addAllPlots{#1}{100}{cycle}{#2}
\nextgroupplot[legend columns=8, legend to name=thisPlotLegend]
\addAllPlots{#1}{100}{star}{#2}
\addlegendentry{DP(Infinity)}
\addlegendentry{DP(1000)}
\addlegendentry{DP(2)}
\addlegendentry{SA}
\addlegendentry{2P}
\addlegendentry{NSGA-II}
\addlegendentry{II}
\addlegendentry{RMQ}
\end{groupplot}

\plotTitle{1}{1}{Chain, 10 tables}
\plotTitle{1}{2}{Chain, 25 tables}
\plotTitle{1}{3}{Chain, 50 tables}
\plotTitle{1}{4}{Chain, 75 tables}
\plotTitle{1}{5}{Chain, 100 tables}

\plotTitle{2}{1}{Cycle, 10 tables}
\plotTitle{2}{2}{Cycle, 25 tables}
\plotTitle{2}{3}{Cycle, 50 tables}
\plotTitle{2}{4}{Cycle, 75 tables}
\plotTitle{2}{5}{Cycle, 100 tables}

\plotTitle{3}{1}{Star, 10 tables}
\plotTitle{3}{2}{Star, 25 tables}
\plotTitle{3}{3}{Star, 50 tables}
\plotTitle{3}{4}{Star, 75 tables}
\plotTitle{3}{5}{Star, 100 tables}
}

\def\analysisPlot#1#2{
\addplot+[unbounded coords=jump] table[header=true,col sep=tab, x index=0, y index=1] {plotsdata/furtherAnalysis/#1M#2.txt};
\addlegendentry{Chain}
\addplot+[unbounded coords=jump] table[header=true,col sep=tab, x index=0, y index=2] {plotsdata/furtherAnalysis/#1M#2.txt};
\addlegendentry{Star}
\addplot+[unbounded coords=jump] table[header=true,col sep=tab, x index=0, y index=3] {plotsdata/furtherAnalysis/#1M#2.txt};
\addlegendentry{Cycle}
}

\def\analysisGroupPlot#1#2#3#4{
\begin{groupplot}[group style={group size=2 by 1, x descriptions at=edge bottom, horizontal sep=1.5cm}, 
width=4cm, 
xlabel=Nr.\ query tables, xlabel near ticks, ylabel near ticks, 
legend to name=analysisLeg, legend columns=3, xtick=data,
ymajorgrids]
\nextgroupplot[ylabel=#1, height=2.75cm]
\analysisPlot{3}{#2}
\nextgroupplot[ylabel=#3, height=2.75cm]
\analysisPlot{3}{#4}
\end{groupplot}
}

\def\comparisonPlot#1#2#3{
\addplot+[unbounded coords=jump, bar shift=-0.1cm, bar width=0.1cm] table[header=true,col sep=tab, x index=0, y index=7] {plotsdata/furtherAnalysis/#1M#3#2.txt};
\addlegendentry{NSGA-II}
\addplot+[unbounded coords=jump, bar shift=0cm, bar width=0.1cm] table[header=true,col sep=tab, x index=0, y index=8] {plotsdata/furtherAnalysis/#1M#3#2.txt};
\addlegendentry{II}
\addplot+[unbounded coords=jump, bar shift=0.1cm, bar width=0.1cm] table[header=true,col sep=tab, x index=0, y index=9] {plotsdata/furtherAnalysis/#1M#3#2.txt};
\addlegendentry{RMQ}
}

\newlength{\comparisonPlotHeight}
\setlength{\comparisonPlotHeight}{2.75cm}

\def\comparisonGroupPlot#1#2#3{
\begin{groupplot}[group style={group size=2 by 2, x descriptions at=edge bottom, horizontal sep=1.5cm}, width=4cm, height=\comparisonPlotHeight,
xlabel=Nr.\ query tables, xlabel near ticks, ylabel near ticks, 
ymode=#3, ybar, xmin=0, xmax=6,
legend to name=comparisonLeg, legend columns=3,
xtick=data, xticklabels={10, 25, 50, 75, 100}, ymajorgrids]
\nextgroupplot[title={2 Metrics, Chain}, ylabel=#1, height=\comparisonPlotHeight]
\comparisonPlot{2}{#2}{chain}
\nextgroupplot[title={3 Metrics, Chain}, ylabel=#1, height=\comparisonPlotHeight]
\comparisonPlot{3}{#2}{chain}
\nextgroupplot[title={2 Metrics, Star}, ylabel=#1, height=\comparisonPlotHeight]
\comparisonPlot{2}{#2}{star}
\nextgroupplot[title={3 Metrics, Star}, ylabel=#1, height=\comparisonPlotHeight]
\comparisonPlot{3}{#2}{star}
\end{groupplot}
}

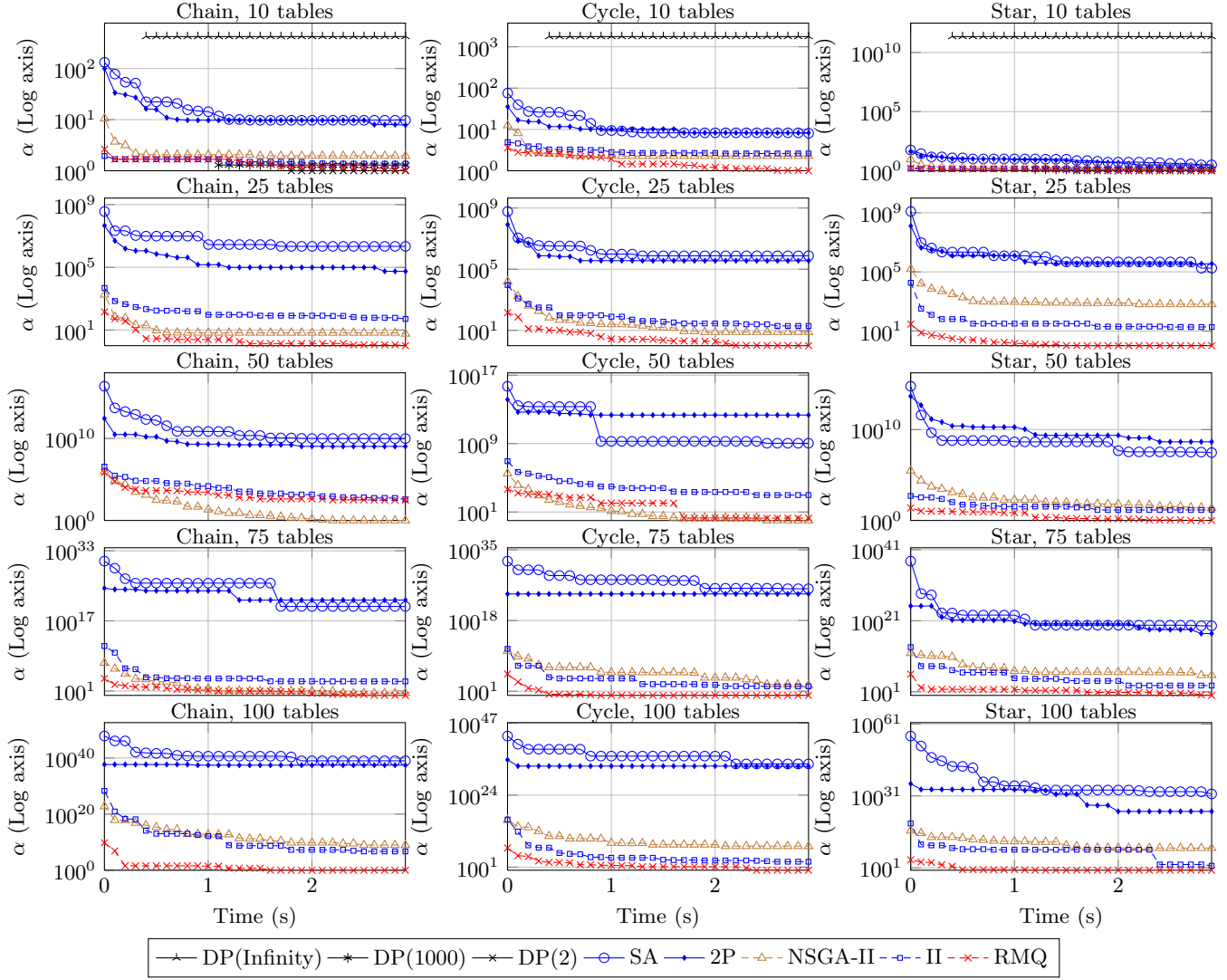
\begin{figure*}[t!]
\centering
\begin{tikzpicture}
\performancePlot{2}{MN}
\end{tikzpicture}
\ref{thisPlotLegend}
\caption{Median of approximation error for two cost metrics as a function of optimization time.\label{twoMetricsFig}}
\end{figure*}

\begin{figure*}[t!]
\centering
\begin{tikzpicture}
\performancePlot{3}{MN}
\end{tikzpicture}
\ref{thisPlotLegend}
\caption{Median of approximation error for three cost metrics as a function of optimization time.\label{threeMetricsFig}}
\end{figure*}
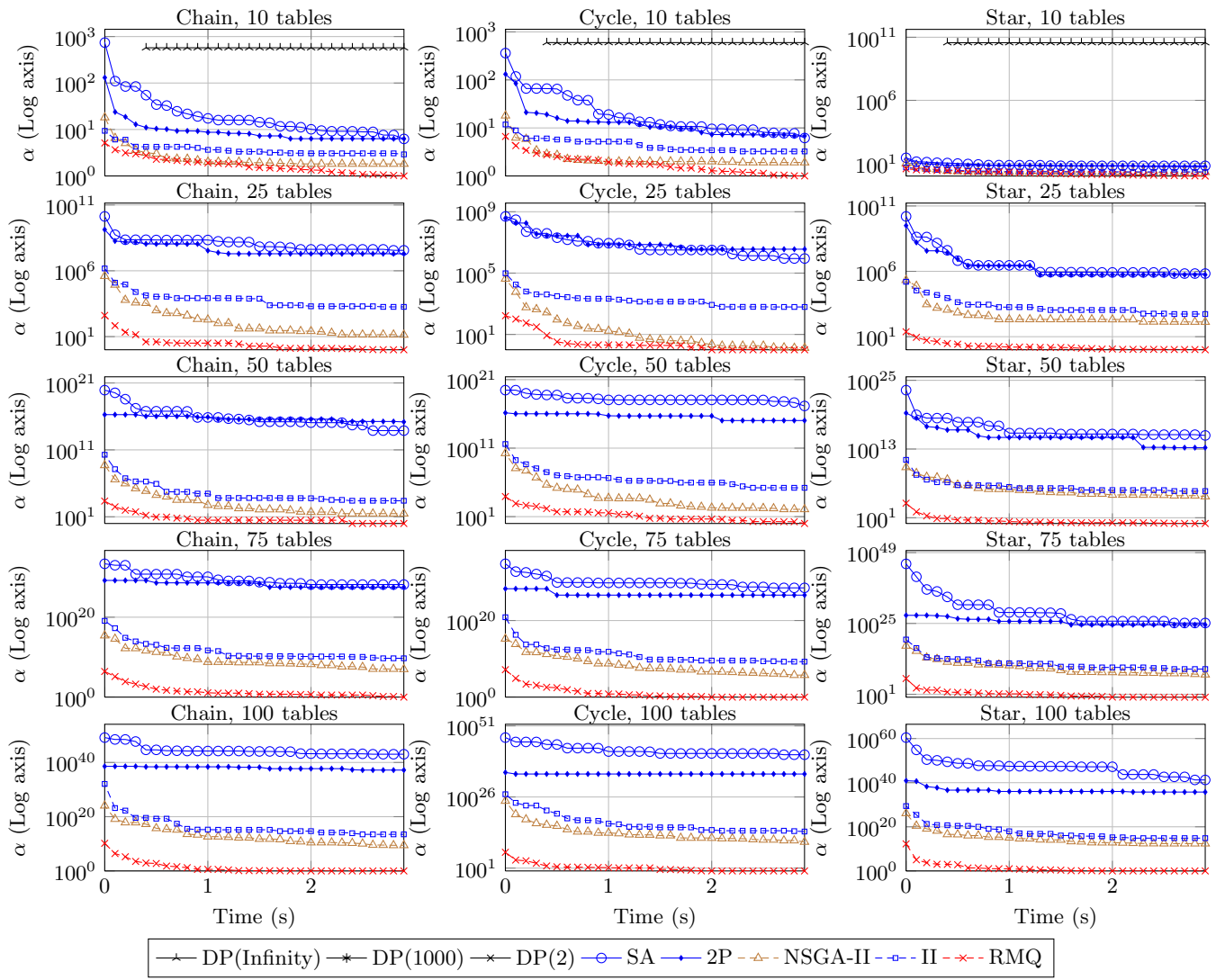

\section{Experimental Evaluation}
\label{experimentsSec}

We describe our experimental setup in Section~\ref{experimentalSetupSub} and discuss the results in Section~\ref{experimentalResultsSub}.

\subsection{Experimental Setup}
\label{experimentalSetupSub}

We compare our RMQ algorithm against other algorithms for multi-objective query optimization. We compare algorithms in terms of how well they approximate the Pareto frontier for a given query after a certain amount of optimization time. We measure the approximation quality in regular intervals during optimization to compare algorithms in different time intervals. This allows to identify algorithms that quickly find reasonable solutions as well as algorithms that take longer to produce reasonable solutions but yield a better approximation in the end.

We judge the set of query plans produced by a certain algorithm by the lowest approximation factor $\alpha$ such that the produced plan set is an $\alpha$-approximate Pareto plan set. A lower $\alpha$ means a better approximation of the real Pareto frontier. This quality metric is equivalent to the $\varepsilon$ metric that was recommended in a seminal paper by Zitzler and Thiele~\cite{Zitzler2003} (setting $\alpha=1+\varepsilon$). Choosing that metric also makes our comparison fair since the dynamic programming based approximation schemes~\cite{Trummer2014} against which we compare have been developed for that metric. As it is common in the area of multi-objective optimization, we often deal with test cases where finding the full set of Pareto solutions is computationally infeasible. We therefore compare the output of each algorithm against an approximation of the real Pareto frontier that is obtained by running all algorithms that we describe in the following for three seconds and taking the union of the obtained result plans.

We compare RMQ against two classes of algorithms: dynamic programming based approximation schemes~\cite{Trummer2014} and generalizations of randomized algorithms that have been proposed for single-objective query optimization~\cite{Steinbrunn1997}. The approximation schemes guarantee to produce a Pareto frontier approximation whose $\alpha$ value does not exceed a user-specified threshold. We denote by DP($\alpha$) the approximation scheme with threshold $\alpha$. Choosing a higher value for $\alpha$ decreases optimization time but choosing a lower value leads to better result quality. We report results for different values of $\alpha$ in the following. 

We experiment with four randomized algorithms (in addition to RMQ itself). By II we denote a generalization of iterative improvement~\cite{Steinbrunn1997} in which we iteratively walk towards local Pareto optima in the search space starting from random query plans. By SA we denote a generalization of the SAIO variant of simulated annealing, described by Steinbrunn et al.~\cite{Steinbrunn1997}. The original algorithm uses the difference between the scalar cost value of the current plan and the cost of a randomly selected neighbor to decide whether to move towards the neighbor plan, based additionally on the current temperature. Our generalization uses the average cost difference between the current plan and its neighbor, averaging over all cost metrics. By 2P we abbreviate the two phase optimization algorithm for query optimization~\cite{Steinbrunn1997}. It executes the II algorithm in a first phase and continues with SA in a second phase. We switch to the second phase after ten iterations of II~\cite{Steinbrunn1997} and choose the initial temperature for SA as described by Steinbrunn et al.~\cite{Steinbrunn1997}. By NSGA-II we abbreviate the Non-Dominated Sort Genetic Algorithm II~\cite{Deb2002}, a genetic algorithm for multi-objective query optimization that has been very successful for scenarios with the same number of cost metrics that we consider.

All our algorithms that are based on local search use the plan transformations for bushy query plans that were described by Steinbrunn et al.~\cite{Steinbrunn1997}. All algorithms using hill climbing (II and 2P) use the same efficient climbing function (see Algorithm~\ref{ParetoClimbAlg}) as our own algorithm. NSGA-II uses an ordinal plan encoding and a corresponding single-point crossover~\cite{Steinbrunn1997}. Our implementation of NSGA-II follows closely the pseudo-code given in the original paper~\cite{Deb2000} and we use the same settings for mutation and crossover probabilities. We use populations of size 200. We experimented with several configurations for each of the randomized algorithms (e.g., experimenting with different population sizes for NSGA-II and different number of examined neighbors for II, SA, and 2P) and report for each algorithm only the configuration that led to optimal performance. 

We generate random queries with a given number of tables in the same way as in prior evaluations of query optimization algorithms~\cite{Steinbrunn1997, Bruno}: we experiment with different join graph structures and use stratified sampling to pick table cardinalities, using the same distribution as Steinbrunn et al.~\cite{Steinbrunn1997}. We consider up to three cost metrics on query plans that were already used for other experimental evaluations in the area of multi-objective query optimization~\cite{Trummer2014}: query execution time, buffer space consumption, and disc space consumption. We report in the following plots median values from 20 test cases per data point. For less than three cost metrics, we select the specified number of cost metrics with uniform distribution from the total set of metrics for each test case.

All algorithms were implemented in Java~1.7 and executed using the Java HotSpot(TM) 64-Bit Server Virtual Machine version on an iMac with  i5-3470S 2.90GHz CPU and 16~GB of DDR3 RAM. We run each algorithm consecutively on all test cases after forcing garbage collection and after a ten seconds code warmup on randomly selected test cases in order to benchmark steady state performance\footnote{\url{http://www.ibm.com/developerworks/library/j-benchmark1/}}.

\subsection{Experimental Results}
\label{experimentalResultsSub}

Figure~\ref{twoMetricsFig} reports results for two cost metrics and Figure~\ref{threeMetricsFig} reports our results for three cost metrics. We report separate results for different query sizes (measured by the number of tables being joined) and join graph structures (chain, cycle, and star shape). We allow up to three seconds of optimization time which might seem excessive for single-objective query optimization but is well within the range that multi-objective query optimizers are typically evaluated on~\cite{Trummer2014, Trummer2015a}. The experiments for producing the data shown in Figures~\ref{twoMetricsFig} and \ref{threeMetricsFig} took around eight hours of optimization time.

Dynamic programming based approximation schemes for multi-objective query optimization have so far been evaluated only on queries joining up to around 10 tables~\cite{Trummer2014, Trummer2015, Trummer2015a}. For 10 table queries, DP(2) finds the best frontier approximation among all evaluated algorithms if two cost metrics are considered. For three cost metrics, DP(2) cannot finish optimization within the given time frame. For queries joining 25 tables and more, none of the approximation schemes finishes optimization within the given time frame. Note that the optimization time required by the approximation schemes we compare against constitutes a lower bound for the optimization time required by the incremental approximation algorithm proposed in our prior work~\cite{Trummer2015a} for reaching the same approximation quality. As the non-incremental approximation schemes cannot finish optimization within the given time frame, even when setting $\alpha=\infty$, the incremental versions will not be able to do so either. This is not surprising as we reach a query size where randomized algorithms would already be used for single-objective query optimization which is computationally far less expensive than multi-objective query optimization. This shows the need for randomized algorithms for multi-objective query optimization such as RMQ.

Among the randomized algorithms that generalize popular randomized algorithms for single-objective query optimization (II, SA, 2P, and NSGA-II), II and NSGA-II generate the best approximations with a large gap to SA and 2P (note the logarithmic y axis for approximation error). It is interesting that II outperforms 2P unlike in traditional query optimization where the roles are reversed. However, SA and 2P both spend most of their time improving one single query plan (2P after limited initial sampling) and are therefore intrinsically based on the assumption that only one very good plan needs to be found as result. Approximating the Pareto frontier requires however to generate a diverse set of query plans which is better accomplished by the seemingly naive II which starts each iteration with a new random plan. NSGA-II usually performs better than II, SA, and 2P. This is not surprising since NSGA-II is a popular algorithm for multi-objective optimization with a moderate number of cost metrics and genetic algorithms have been shown to perform well for classical query optimization~\cite{Bennett1991}.

Our own algorithm, RMQ, outperforms all other algorithms in the majority of cases. The gap to the other algorithms increases in the number of query tables and in the number of cost metrics. For two cost metrics (see Figure~\ref{twoMetricsFig}), RMQ is competitive and often significantly better than all other algorithms over the entire optimization time period starting from more than 50 join tables. For star-shaped query graphs, RMQ is better than competing algorithms starting from 25 join tables already. Considering three cost metrics (see Figure~\ref{threeMetricsFig}) increases the gap between RMQ and all other algorithms. Starting from 25 join tables, RMQ dominates over the entire optimization time period. The gap in terms of approximation error reaches many orders of magnitude for large queries.

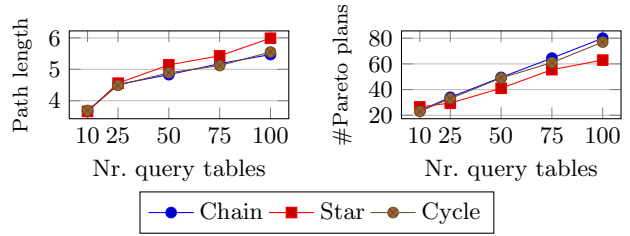
\begin{figure}[t!]
\centering
\begin{tikzpicture}
\analysisGroupPlot{{\small Path length}}{pathLength}{{\small \#Pareto plans}}{nrParetos}
\end{tikzpicture}

\ref{analysisLeg}
\caption{Median path length from random plan to next local Pareto optimum and median number of Pareto plans found by RMQ for three cost metrics.\label{pathParetoFig}}
\end{figure}

Figures~\ref{pathParetoFig} shows additional statistics: on the left side we see that the average path length from a random plan to the nearest Pareto optimum (using function~\textproc{ParetoClimb}) grows slowly in the number of query tables as postulated in our formal analysis. On the right side, we see that the number of Pareto plans grows in the number of query tables which corroborates prior results~\cite{Trummer2014}. This tendency explains why the approximation error of randomized algorithms increases in the query size: having more Pareto plans makes it harder to obtain a good approximation. 


We summarize the main results of our experimental evaluation. Dynamic programming based algorithms for multi-objective query optimization are only applicable for small queries. Using randomized algorithms, we were able to approximate the Pareto frontier for large queries joining up to 100 tables. The algorithm proposed in this paper performs best among the randomized algorithms over a broad range of scenarios. Among the remaining randomized algorithms, a general-purpose genetic algorithm for multi-objective optimization performs best, followed closely by an iterative improvement algorithm that uses the efficient hill climbing function that was introduced in this paper as well. 


So far we have compared algorithms for a few seconds of optimization time. It is generally advantageous to minimize optimization time as it adds to execution time. In the context of multi-objective query optimization, it is even more important as query optimization can be an interactive process in which users have to wait until optimization finishes~\cite{Trummer2015a}. We have however also compared algorithms for up to 30 seconds of optimization time and present the corresponding results in the appendix. Until now we compared algorithms in terms of how well they approximate an approximated Pareto frontier since calculating the real Pareto frontier would lead to prohibitive computational costs for many of the query sizes we consider. We show results for small queries where we calculated the real Pareto frontier in the appendix. Furthermore, the appendix contains results for different query generation methods. The main results of our experiments are stable across all scenarios.

\section{Conclusion}

The applicability of multi-objective query optimization has so far been severely restricted by the fact that all available algorithms have exponential complexity in the number of query tables. We presented, analyzed, and evaluated the first polynomial time heuristic for multi-objective query optimization. We have shown that our algorithm scales to queries that are by one order of magnitude larger than the ones prior multi-objective query optimizers were so far evaluated on. We envision our algorithm being used in future multi-objective query optimizers that apply dynamic programming based algorithms for small queries and switch to randomized algorithms starting from a certain number of query tables, similar to single-objective optimizers today.

\section{Acknowledgments}

This work was supported by ERC grant 279804. I. Trummer is supported by a Google European PhD Fellowship.

\bibliographystyle{abbrv}
\bibliography{../../library}

\appendix

The performance of query optimization algorithms may depend on the probability distribution over predicate selectivity values used during random query generation. For our experiments in Section~\ref{experimentsSec}, we used the original method proposed by Steinbrunn et al.~\cite{Steinbrunn1997} to select the selectivity of join predicates. We performed an additional series of experiments in which we select predicate selectivity according to the MinMax method proposed by Bruno instead~\cite{Bruno}. Using that method, each join has an output cardinality between the cardinalities of the two input relations. The purpose of those additional experiments is to verify whether our experimental results from Section~\ref{experimentsSec} generalize. We report the results of the second series of experiments in Figures~\ref{twoMetricsMinMaxFig} and \ref{threeMetricsMinMaxFig}. The experimental setup is the same as described in Section~\ref{experimentalSetupSub} except for the choice of selectivity values. We focus on queries joining between 25 to 100 tables where randomized algorithms become better than dynamic programming variants.

The results are largely consistent with the results we obtained in our first series of experiments. Our own algorithm outperforms all other approaches significantly for large queries and many cost metrics, in particular during the first second of optimization time. NSGA-II performs well for smaller queries while its approximation error is by many orders of magnitude sub-optimal for large queries. II comes close to NSGA-II in certain scenarios while the two randomized algorithms based on simulated annealing, SA and 2P, perform badly. The approximation schemes do not scale to queries of 25 tables and more.

\def\performancePlotMinMax#1#2{
\begin{groupplot}[group style={group size=3 by 4, x descriptions at=edge bottom, horizontal sep=1.5cm, vertical sep=0.4cm}, 
width=6cm, height=3.25cm,
xlabel=Time (s), ylabel=$\alpha$ (Log axis), xlabel near ticks, ylabel near ticks,
ymode=log, xmajorgrids, xmin=0, xmax=2.9, ymin=1,
log basis y=100, ymajorgrids]
\nextgroupplot
\addAllPlots{#1}{25}{chain}{#2}
\nextgroupplot
\addAllPlots{#1}{25}{cycle}{#2}
\nextgroupplot
\addAllPlots{#1}{25}{star}{#2}
\nextgroupplot
\addAllPlots{#1}{50}{chain}{#2}
\nextgroupplot
\addAllPlots{#1}{50}{cycle}{#2}
\nextgroupplot
\addAllPlots{#1}{50}{star}{#2}
\nextgroupplot
\addAllPlots{#1}{75}{chain}{#2}
\nextgroupplot
\addAllPlots{#1}{75}{cycle}{#2}
\nextgroupplot
\addAllPlots{#1}{75}{star}{#2}
\nextgroupplot
\addAllPlots{#1}{100}{chain}{#2}
\nextgroupplot
\addAllPlots{#1}{100}{cycle}{#2}
\nextgroupplot[legend columns=8, legend to name=thisPlotLegend]
\addAllPlots{#1}{100}{star}{#2}
\addlegendentry{DP(Infinity)}
\addlegendentry{DP(1000)}
\addlegendentry{DP(2)}
\addlegendentry{SA}
\addlegendentry{2P}
\addlegendentry{NSGA-II}
\addlegendentry{II}
\addlegendentry{RMQ}
\end{groupplot}

\plotTitle{1}{1}{Chain, 25 tables}
\plotTitle{1}{2}{Chain, 50 tables}
\plotTitle{1}{3}{Chain, 75 tables}
\plotTitle{1}{4}{Chain, 100 tables}

\plotTitle{2}{1}{Cycle, 25 tables}
\plotTitle{2}{2}{Cycle, 50 tables}
\plotTitle{2}{3}{Cycle, 75 tables}
\plotTitle{2}{4}{Cycle, 100 tables}

\plotTitle{3}{1}{Star, 25 tables}
\plotTitle{3}{2}{Star, 50 tables}
\plotTitle{3}{3}{Star, 75 tables}
\plotTitle{3}{4}{Star, 100 tables}
}

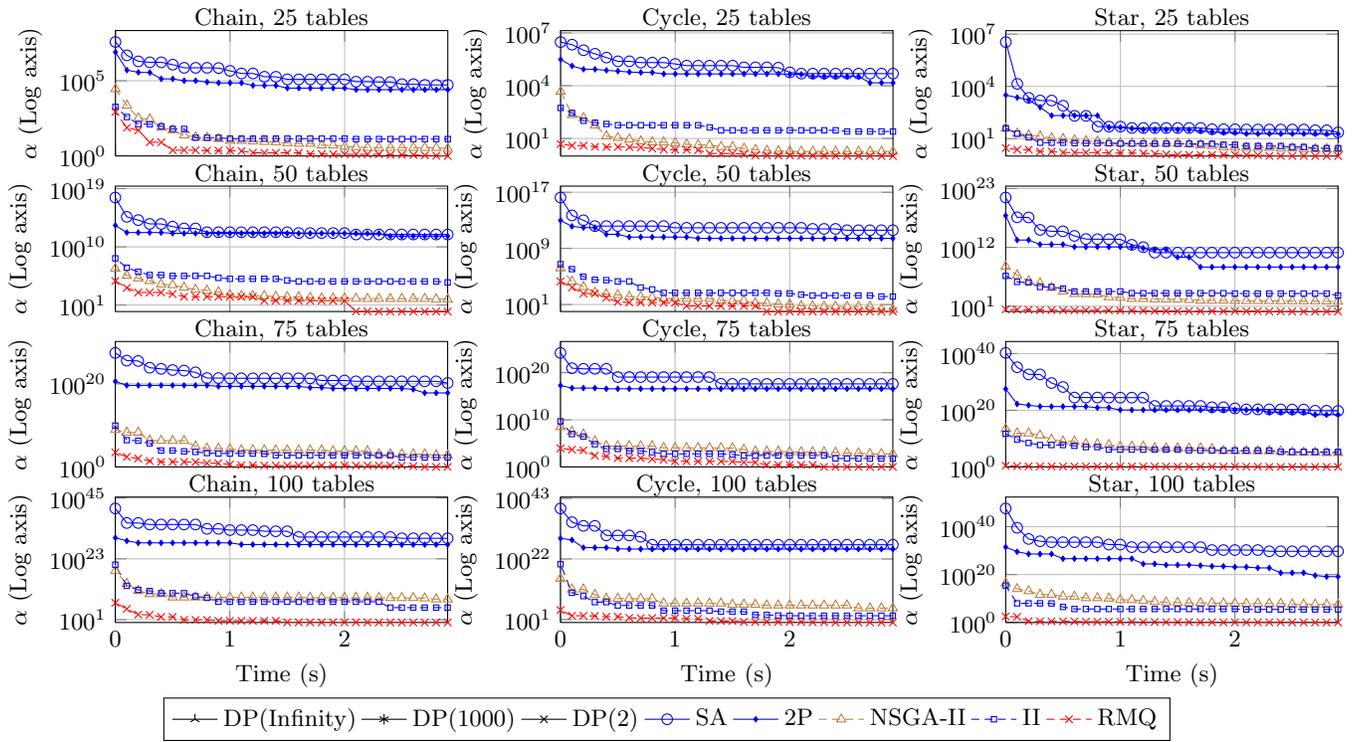
\begin{figure*}[t!]
\centering
\begin{tikzpicture}[declare function={Infinity=inf;}]
\performancePlotMinMax{2}{MinMax}
\end{tikzpicture}
\ref{thisPlotLegend}
\caption{Median of approximation error for two cost metrics as a function of optimization time (MinMax joins).\label{twoMetricsMinMaxFig}}
\end{figure*}

\begin{figure*}[t!]
\centering
\begin{tikzpicture}[declare function={Infinity=inf;}]
\performancePlotMinMax{3}{MinMax}
\end{tikzpicture}
\ref{thisPlotLegend}
\caption{Median of approximation error for three cost metrics as a function of optimization time (MinMax joins).\label{threeMetricsMinMaxFig}}
\end{figure*}
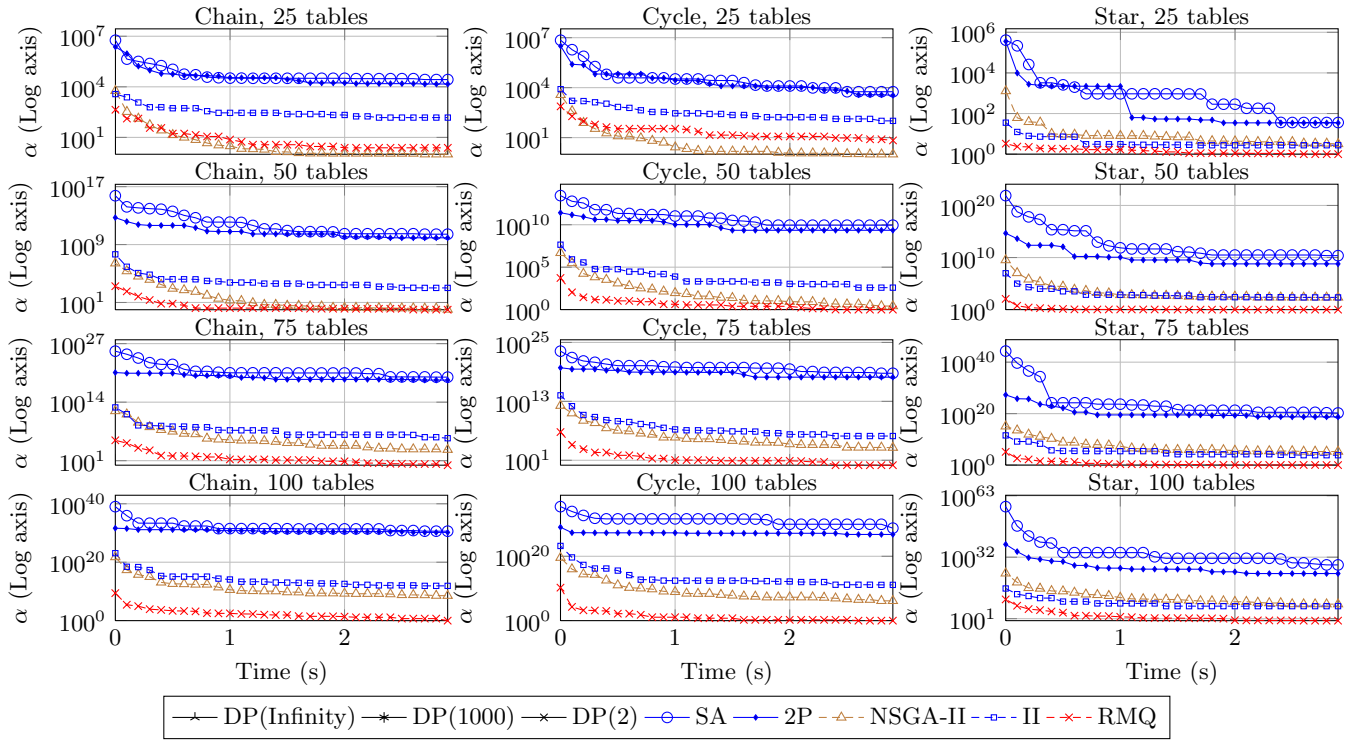

\def\addAllPlotsSpecial#1#2#3#4#5{
\addplot+[unbounded coords=jump, mark=Mercedes star, draw=black] table[header=true,col sep=tab, x index=0, y index=1] {plotsdata/#5/scaled/#1M#2T#3#4.txt};
\addplot+[unbounded coords=jump, mark=asterisk, draw=black] table[header=true,col sep=tab, x index=0, y index=2] {plotsdata/#5/scaled/#1M#2T#3#4.txt};
\addplot+[unbounded coords=jump, mark=x, draw=black] table[header=true,col sep=tab, x index=0, y index=3] {plotsdata/#5/scaled/#1M#2T#3#4.txt};
\addplot+[unbounded coords=jump, mark=o, mark size=2, draw=blue] table[header=true,col sep=tab, x index=0, y index=4] {plotsdata/#5/scaled/#1M#2T#3#4.txt};
\addplot+[unbounded coords=jump, mark=diamond*, mark size=1, draw=blue] table[header=true,col sep=tab, x index=0, y index=6] {plotsdata/#5/scaled/#1M#2T#3#4.txt};
\addplot+[unbounded coords=jump, mark=triangle, mark size=2, draw=brown] table[header=true,col sep=tab, x index=0, y index=7] {plotsdata/#5/scaled/#1M#2T#3#4.txt};
\addplot+[unbounded coords=jump, mark=square, mark size=1, draw=blue] table[header=true,col sep=tab, x index=0, y index=8] {plotsdata/#5/scaled/#1M#2T#3#4.txt};
\addplot+[unbounded coords=jump, mark=x, mark size=2, draw=red] table[header=true,col sep=tab, x index=0, y index=9] {plotsdata/#5/scaled/#1M#2T#3#4.txt};
}

\def\performancePlotLong#1#2{
\begin{groupplot}[group style={group size=3 by 2, x descriptions at=edge bottom, horizontal sep=1.5cm, vertical sep=0.4cm}, 
width=6cm, height=3.25cm,
xlabel=Time (s), ylabel=$\alpha$ (Log axis), xlabel near ticks, ylabel near ticks,
ymode=log, xmajorgrids, xmin=0, xmax=29, ymin=1, ymax=10E10,
log basis y=10, ymajorgrids, ytickten={2,4,6,8,10,12}]
\nextgroupplot
\addAllPlotsSpecial{#1}{50}{chain}{#2}{longRuns}
\nextgroupplot
\addAllPlotsSpecial{#1}{50}{cycle}{#2}{longRuns}
\nextgroupplot
\addAllPlotsSpecial{#1}{50}{star}{#2}{longRuns}
\nextgroupplot
\addAllPlotsSpecial{#1}{100}{chain}{#2}{longRuns}
\nextgroupplot
\addAllPlotsSpecial{#1}{100}{cycle}{#2}{longRuns}
\nextgroupplot[legend columns=8, legend to name=thisPlotLegend]
\addAllPlotsSpecial{#1}{100}{star}{#2}{longRuns}
\addlegendentry{DP(Infinity)}
\addlegendentry{DP(1000)}
\addlegendentry{DP(2)}
\addlegendentry{SA}
\addlegendentry{2P}
\addlegendentry{NSGA-II}
\addlegendentry{II}
\addlegendentry{RMQ}
\end{groupplot}

\plotTitle{1}{1}{Chain, 50 tables}
\plotTitle{1}{2}{Chain, 100 tables}

\plotTitle{2}{1}{Cycle, 50 tables}
\plotTitle{2}{2}{Cycle, 100 tables}

\plotTitle{3}{1}{Star, 50 tables}
\plotTitle{3}{2}{Star, 100 tables}
}

\begin{figure*}[t!]
\centering
\begin{tikzpicture}[declare function={Infinity=inf;}]
\performancePlotLong{2}{MN}
\end{tikzpicture}
\ref{thisPlotLegend}
\caption{Median of approximation error in the interval $[1,10^{10}]$ for two cost metrics and up to 30 seconds of optimization time.\label{twoMetricsLongFig}}
\end{figure*}
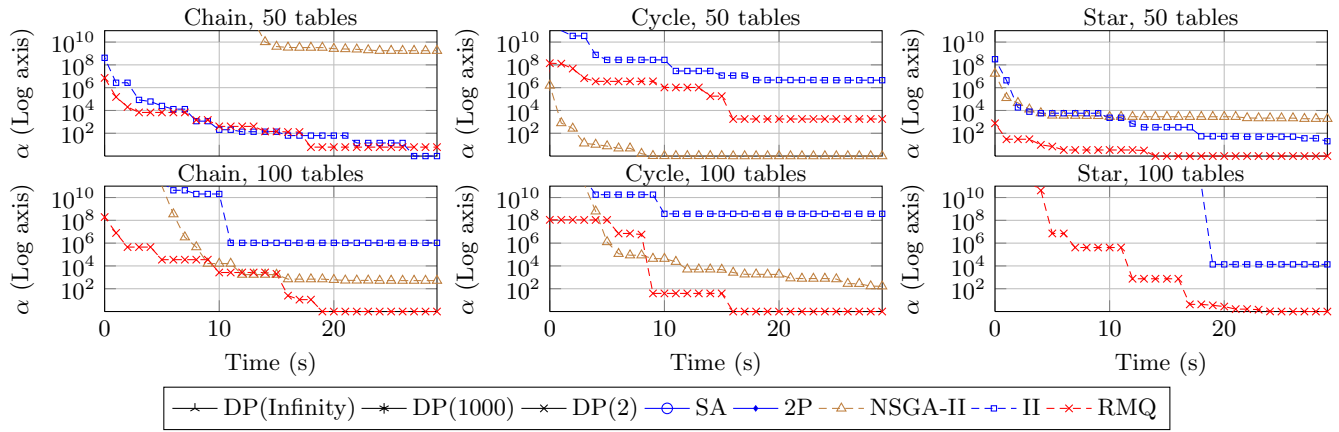

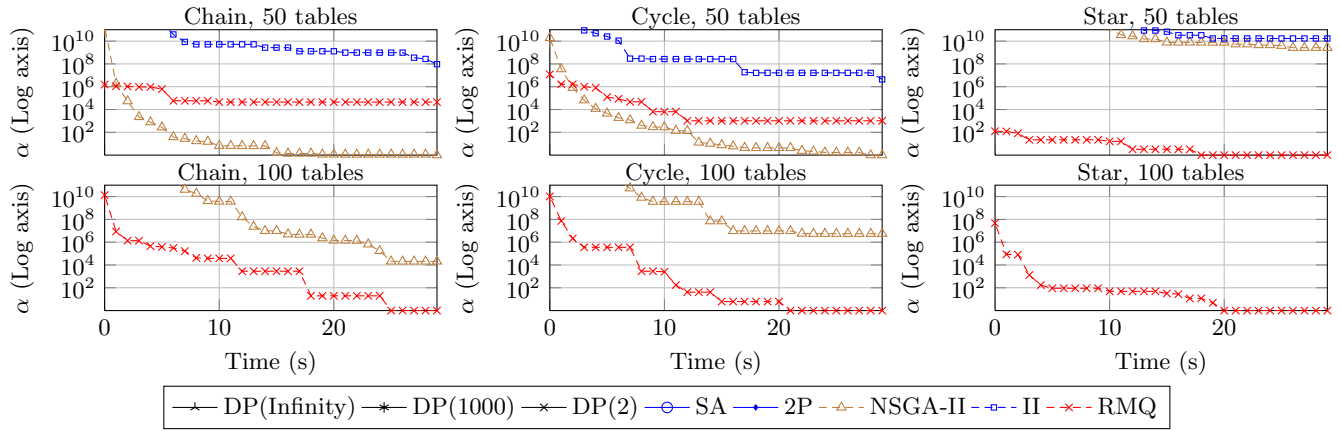
\begin{figure*}[t!]
\centering
\begin{tikzpicture}[declare function={Infinity=inf;}]
\performancePlotLong{3}{MN}
\end{tikzpicture}
\ref{thisPlotLegend}
\caption{Median of approximation error in the interval $[1,10^{10}]$ for three cost metrics and up to 30 seconds of optimization time.\label{threeMetricsLongFig}}
\end{figure*}

Multi-objective query optimization needs to integrate user preferences in order to determine the optimal query plan. One possibility to integrate user preferences is to present an approximation of the plan Pareto frontier for a given query to the user such that the user can select the preferred cost tradeoff~\cite{Trummer2015a}. This means that users have to wait after submitting a query until optimization finishes in order to make their selection. In that scenario, an optimization time of only a few seconds is desirable. If users formalize their preferences before optimization starts (e.g., in the form of cost weights and cost bounds~\cite{Trummer2014}) then longer optimization times can be acceptable. The second scenario motivates us to compare the query optimization algorithms for longer periods of optimization time.

We executed another series of experiments giving each algorithm 30 seconds of optimization time. We reduced the number of test cases per scenario from 20 to 10 and only experimented with two query sizes in order to decrease the time overhead for the experiments. Figures~\ref{twoMetricsLongFig} and \ref{threeMetricsLongFig} report the development of the median approximation error over 30 seconds of optimization time. We restrict the y domain of the figures and only show errors up to $\alpha=10^{10}$. Using that thresholds allows us to visualize performance difference between the well performing algorithms that would otherwise become indistinguishable even though they are significant. 
Albeit we ran the same algorithms as in the previous plots on all test cases, the plots only show data points for a subset of those algorithms. For the dynamic programming variants, the reason for not being represented in the plots is that they did not return any results even within 30 seconds of optimization time. For simulated annealing and two-phase optimization, the reason for not being represented in the plots is that their approximation error is significantly above the threshold of $10^{10}$ (often more than $10^{110}$). 

The tendencies remain the same: our randomized algorithm, the genetic algorithm, and iterative improvement with our fast climbing function are the best randomized algorithms. RMQ is usually better than iterative improvement. For queries with up to 50 tables, it depends on the number of cost metrics and the join graph structure whether RMQ or the genetic algorithm perform better. For more than 50 tables, RMQ outperforms all other algorithms over most of the optimization time period, the margin increases in the number of plan cost metrics.

\def\performancePlotExact#1#2{
\begin{groupplot}[group style={group size=3 by 2, x descriptions at=edge bottom, horizontal sep=1.5cm, vertical sep=0.4cm}, 
width=6cm, height=3.25cm,
xlabel=Time (s), ylabel=$\alpha$ (Log axis), xlabel near ticks, ylabel near ticks,
ymode=log, xmajorgrids, ymajorgrids, xmin=0, xmax=29, ymin=1, ymax=2,
ytick={1,1.2,1.4,1.6,1.8,2}, yticklabels={1,1.2,1.4,1.6,1.8,2}]
\nextgroupplot
\addAllPlotsSpecial{#1}{4}{chain}{#2}{exactRuns}
\nextgroupplot
\addAllPlotsSpecial{#1}{4}{cycle}{#2}{exactRuns}
\nextgroupplot
\addAllPlotsSpecial{#1}{4}{star}{#2}{exactRuns}
\nextgroupplot
\addAllPlotsSpecial{#1}{8}{chain}{#2}{exactRuns}
\nextgroupplot
\addAllPlotsSpecial{#1}{8}{cycle}{#2}{exactRuns}
\nextgroupplot[legend columns=8, legend to name=thisPlotLegend]
\addAllPlotsSpecial{#1}{8}{star}{#2}{exactRuns}
\addlegendentry{DP(Infinity)}
\addlegendentry{DP(1000)}
\addlegendentry{DP(2)}
\addlegendentry{SA}
\addlegendentry{2P}
\addlegendentry{NSGA-II}
\addlegendentry{II}
\addlegendentry{RMQ}
\end{groupplot}

\plotTitle{1}{1}{Chain, 4 tables}
\plotTitle{1}{2}{Chain, 8 tables}

\plotTitle{2}{1}{Cycle, 4 tables}
\plotTitle{2}{2}{Cycle, 8 tables}

\plotTitle{3}{1}{Star, 4 tables}
\plotTitle{3}{2}{Star, 8 tables}
}

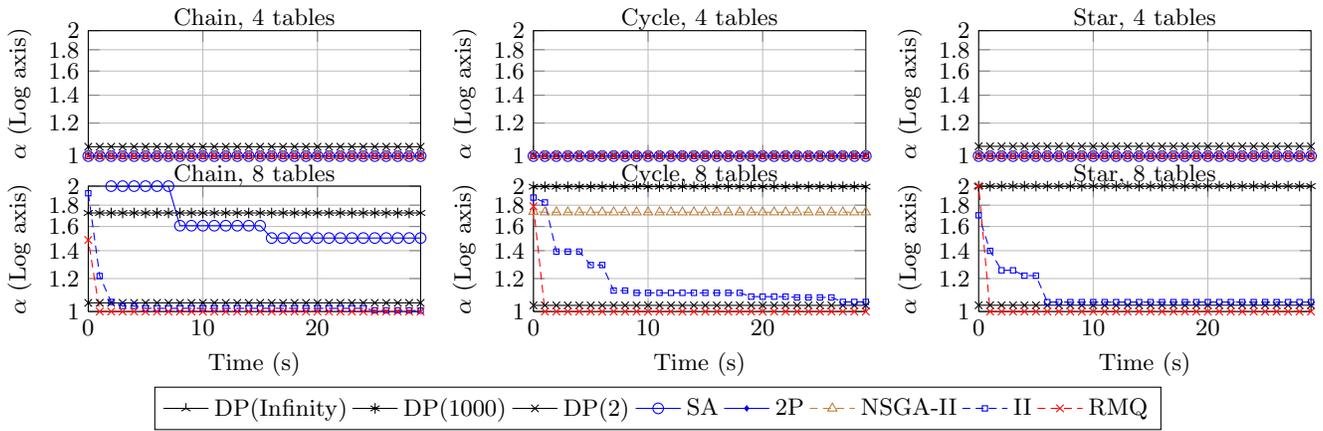
\begin{figure*}[t!]
\centering
\begin{tikzpicture}[declare function={Infinity=inf;}]
\performancePlotExact{2}{MN}
\end{tikzpicture}
\ref{thisPlotLegend}
\caption{Median of precise approximation error in the interval $[1,2]$ for small queries and two cost metrics.\label{twoMetricsExactFig}}
\end{figure*}

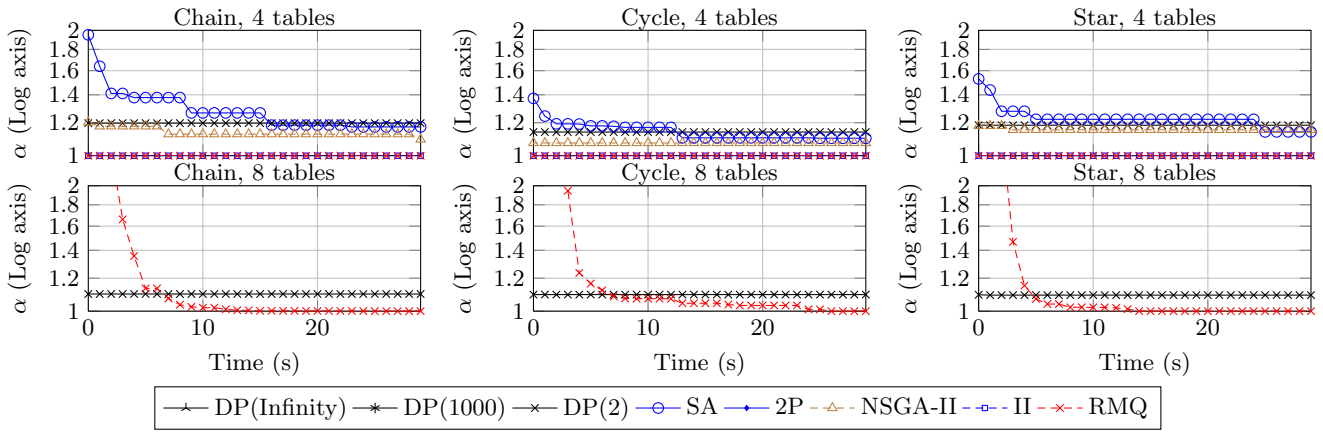
\begin{figure*}[t!]
\centering
\begin{tikzpicture}[declare function={Infinity=inf;}]
\performancePlotExact{3}{MN}
\end{tikzpicture}
\ref{thisPlotLegend}
\caption{Median of precise approximation error in the interval $[1,2]$ for small queries and three cost metrics.\label{threeMetricsExactFig}}
\end{figure*}

We have finally run an additional test series in which we calculate the approximation error more precisely than in the previous experiments. For large queries, we have no choice but to evaluate approximation precision with regards to an approximated Pareto frontier that is generated by the same randomized algorithms that we evaluate. For small queries, we can use the dynamic programming based approximation schemes to calculate an approximated Pareto frontier with formal guarantees on how closely it is approximated. For the last series of experiments, we calculate the approximated Pareto frontier by the approximation scheme, setting $\alpha=1.01$. This means that the calculated approximation error is guaranteed to be precise within a very small tolerance. We restrict ourselves to small queries joining between four and eight tables. Note that the size of the Pareto frontier produced by the approximation scheme reaches several hundreds of Pareto plans already for such small queries.

Figures~\ref{twoMetricsExactFig} and \ref{threeMetricsExactFig} show the results. We allow again 30 seconds of optimization time and restrict the plots to the domain $[1,2]$ for the approximation error. This has the same implications as discussed in the previous paragraphs: algorithms with exceedingly high errors are not shown in order not to obfuscate the performance differences between the competitive algorithms. We are primarily interested in whether or not the randomized algorithms converge to the reference Pareto frontier for small queries; therefore we choose to show a small error range above $\alpha=1$.

We find that our algorithms converges in average to a perfect approximation with $\alpha=1$ while this is not always the case for the other algorithms. In particular for queries joining eight tables and three plan cost metrics, our algorithm is the only one among all randomized algorithms that achieves a perfect approximation. The dynamic programming based approximation scheme with $\alpha=2$ performs well for small queries. It generates output nearly immediately and the approximation error is much lower than the theoretical worst case bound. The approximation scheme configuration using $\alpha=1.01$, the one generating the reference frontier, produces its solutions after less than two seconds of optimization time in average, even for three cost metrics and eight tables. We do not show results for that algorithm in the plots as its approximation error is minimal by definition. We conclude that approximation schemes often outperform randomized algorithms for small queries. 

\balance


\end{document}